\documentclass[submission,copyright,creativecommons]{eptcs}
\providecommand{\event}{PLACES 2019} % Name of the event you are submitting to
\usepackage{underscore}           % Only needed if you use pdflatex.

\usepackage{booktabs}   %% For formal tables:
\usepackage{subcaption} %% For complex figures with subfigures/subcaptions
\usepackage{mathtools}
\usepackage{proof-dashed}
\usepackage{tikz-cd}
\usepackage{latexsym}
\usepackage{amssymb}            % for \multimap (-o)
\usepackage{stmaryrd}           % for \binampersand (&), \bindnasrepma (\paar)
\usepackage{color}
\newcommand{\m}[1]{\mathsf{#1}}

\newcommand{\DD}{\mathcal{D}}
\newcommand{\EE}{\mathcal{E}}

\newcommand{\ih}[1]{\mbox{i.h.}(#1)}

\newcommand{\mi}[1]{\mbox{\it #1}}

\newcommand{\CC}{\mathcal{C}}

\newcommand{\up}{{\uparrow}}
\newcommand{\down}{{\downarrow}}
\newcommand{\uup}[2]{{{}^{#1}_{#2}}{{\uparrow}\kern-0.2em{\downarrow}}}
\newcommand{\ddown}[2]{{{}^{#1}_{#2}}{{\downarrow}\kern-0.2em{\uparrow}}}
\newcommand{\uscore}{\mbox{\tt\char`\_}}

\newcommand{\mmode}[1]{{\mathchoice{\m{#1}}{\m{#1}}{\scriptscriptstyle\m{#1}}{\scriptscriptstyle\m{#1}}}}
\newcommand{\mL}{\mmode{L}}

\newcommand{\mU}{\mmode{U}}
\newcommand{\mV}{\mmode{V}}
\newcommand{\mX}{\mmode{X}}

\newcommand{\seq}{\vdash}
\newcommand{\sseq}{\vvdash}

\newcommand{\semi}{\mathrel{;}}

\newcommand{\cut}{\m{cut}}

\newcommand{\id}{\m{id}}

\newcommand{\vvdash}{\mathrel{\vdash\kern-0.8ex\vdash}}
\newcommand{\weaken}{\m{weaken}}
\newcommand{\contract}{\m{contract}}

\newcommand{\lolli}{\multimap}
\newcommand{\tensor}{\otimes}
\newcommand{\with}{\mathbin{\binampersand}}

\newcommand{\one}{\mathbf{1}}

\newcommand{\bang}{{!}}

\newcommand{\upshift}{\mathord{\uparrow}}
\newcommand{\downshift}{\mathord{\downarrow}}
\newcommand{\case}{\mathsf{case}}

\newcommand*{\inferlabel}[1]{\raisebox{1.5ex}{$#1$}}

\usepackage{amsmath}
\usepackage{amsthm}
\usepackage{enumerate}
\usepackage{float}
\usepackage{xcolor}
\usepackage{cleveref}
\usepackage{url}
\usepackage{breakurl}
\floatstyle{boxed} 
\restylefloat{figure}

\theoremstyle{plain}
\newtheorem{example}{Example}
\newtheorem{theorem}{Theorem}
\newtheorem{lemma}{Lemma}
\newtheorem{corollary}{Corollary}

\theoremstyle{definition}
\newtheorem{definition}{Definition}
\newcommand{\pgraph}[1]{\medskip\noindent\textbf{#1}\quad}

\title{A Message-Passing Interpretation of Adjoint Logic}
\author{Klaas Pruiksma
\institute{Carnegie Mellon University}
\institute{Computer Science Department}
\email{kpruiksm@cs.cmu.edu}
\and
Frank Pfenning
    \institute{Carnegie Mellon University}
    \institute{Computer Science Department}
    \email{fp@cs.cmu.edu}
}
\def\titlerunning{Message-Passing Adjoint Logic}
\def\authorrunning{K. Pruiksma, F. Pfenning}
\begin{document}
\maketitle

\begin{abstract}
    
We present a system of session types based on \emph{adjoint logic}
which generalize standard binary session types~\cite{Honda93concur}.
Our system allows us to uniformly capture several new behaviors in
the space of asynchronous message-passing communication, including
\emph{multicast}, where a process sends a single message to multiple
clients, \emph{replicable services}, which have multiple clients and replicate
themselves on-demand to handle requests from those clients,
and \emph{cancellation}, where a process discards a channel
without communicating along it. We provide session fidelity and
deadlock-freedom results for this system, from which we then derive
a logically justified form of garbage collection.

% Adjoint logic provides a schematic way to combine multiple logics,
% some of which may be substructural, through modal operators that
% are adjoint to each other. We provide a simple formulation
% for adjoint logic with explicit structural rules. The adjoint logic is
% parameterized by a preorder of modes of truth characterizing (potential)
% dependence between the modes. We demonstrate that suitable choices
% of this preorder allow us to directly embed various logics including lax logic,
% judgmental S4, LNL, and intuitionistic subexponential linear logic
% into adjoint logic.

% Under the proofs-as-programs paradigm, proofs correspond to
% concurrent processes and cut reduction to synchronous
% communication. We show how to restructure the sequent calculus so
% that cut reduction entails asynchronous communication and give an
% operational interpretation that provides session-typed communication
% extended with multicast and distributed garbage collection.
\end{abstract}

\section{Introduction}

\emph{Binary session types}~\cite{Honda93concur} were designed to
specify the communication behavior between two message-passing
processes.  But there are patterns of communication that do not fall
into this category.  One example is one provider of a \emph{replicable
  service} with multiple clients.  Another is a \emph{multicast}, that
is, a process sending one message to multiple recipients.  A third one
is a client that no longer wishes to use a service, a form of
\emph{cancellation}.  In this paper we provide a uniform language and
operational semantics rooted in logic that captures such patterns of
asynchronous communication.  It generalizes the usual binary session
types by supporting multiple \emph{modes of communication}.  In each
of these modes every channel has a unique \emph{provider} (which may
send or receive), and possibly multiple clients.  We identify the
following modes: \emph{linear} (a unique client that must communicate,
as with the usual binary session types), \emph{affine} (a unique
client that may communicate or cancel), \emph{strict} (multiple
clients, each of which must communicate), and \emph{unrestricted}
(multiple clients, each of which may or may not communicate, which
captures both replicable services and multicast).

A type system that uniformly integrates all of these patterns is not
obvious if we want to preserve the desirable properties of session
fidelity and deadlock freedom that we obtain from binary session
types. Underlying our approach is \emph{adjoint
  logic}~\cite{Reed09un,Licata16lfcs,Pruiksma18un}, which generalizes
\emph{intuitionistic linear logic}~\cite{Girard87tcs,Girard87tapsoft}
and LNL~\cite{Benton94csl} by decomposing the usual exponential
modality $!A$ into two adjoint modal operators and also affords
individual control over the structural rules of weakening and
contraction. We provide a formulation of adjoint logic in which cut
reduction corresponds to asynchronous communication, and from which
session fidelity and deadlock freedom derive. Moreover, our
formulation uses a form of explicit structural rules embedded in a
multicut, where weakening corresponds to cancellation and contraction
corresponds to duplication of a message or service.

Some of these patterns have been previously addressed with varying
degrees of proximity to an underlying logic.  A replicable service
with multiple clients can be achieved with \emph{access
  points}~\cite{Gay10jfp} or \emph{persistent services} of type
$!A$~\cite{Caires10concur}. Cancellation can be addressed with affine
types~\cite{Mostrous14coordination,Scalas16ecoop,Padovani17icfp}
further developed for asynchronous communication and general handling
of failure~\cite{Fowler19popl}. Cancellation can also be handled with
modalities used to label cancellable types~\cite{Caires2017linearity}.
This approach differs from ours in a few respects --- first,
Caires and P{\'e}rez work in a purely synchronous setting,
without multicast, and second, they focus heavily on introducing
nondeterminism, which we believe to be orthogonal to
(our form of) cancellation.
Closest to the present proposal is a
polarized formulation of asynchronous communication in adjoint
logic~\cite{Pfenning15fossacs} which had several shortcomings that are
addressed here. Specifically, the mode hierarchy was fixed to have
only three modes (linear, affine, and unrestricted), and the
unrestricted mode only allowed a single kind of proposition
${\up}^\mU_m A_m$. This meant that, for example, multicast was not
representable. Also, the rules left weakening and contraction
implicit, which means that there is no explicit cancellation or
distributed garbage collection, which is only briefly hinted at as a
possibility~\cite{Griffith16phd}.

The Curry-Howard correspondence relates propositions to types, proofs
to programs, and proof reduction to computation. Cut reductions in a
pure sequent calculcus for linear
logic~\cite{Caires10concur,Wadler12icfp} naturally correspond to
synchronous communication because both premises of the cut are reduced
at the same time.  We reformulate adjoint logic with a nonstandard
sequent calculus in which noninvertible rules are presented as axioms,
that is, rules with no premises.  As our operational interpretation
shows, an axiom can be seen as a message and cut reduction in this
sequent calculus corresponds to asynchronous communication.  Another
unusual aspect of our sequent calculus is that we generalize cut to a
sound rule of multicut~\cite{Gentzen35,Negri01book}, which operationally allows
one provider to connect with multiple clients.  Two further
consequences of this reformulation are that (a) no explicit rules are
needed for weakening and contraction, and yet (b) channels and
resources are tracked with sufficient precision that computation in a
network of processes ``leaves no garbage'' (see
\cref{sec:metatheory}).  This is the concurrent realization of the
early observation by Girard and Lafont~\cite{Girard87tapsoft} that
functional computation based on intuitionistic linear logic does not
require a garbage collector.
Cancellation~\cite{Mostrous14coordination,Fowler19popl} is a natural
consequence, without requiring any special mechanism, but our system
goes beyond it in the sense that processes with multiple clients will
also terminate once no clients are left.

We begin with a brief discussion of our type system (\cref{sec:typing}), deferring
discussion of the underlying logic to \cref{app:adjoint-logic}, in order to focus on the programming
system. We then present an operational semantics (\cref{sec:operational}): our first major
contribution. It models a variety of asynchronous communication behaviors, uniformly generalizing
previous systems. We close by briefly presenting our results on
session fidelity and deadlock-freedom, along with a brief discussion of the ``garbage-collection''
result that follows from them (\cref{sec:metatheory}).

\section{Language and Typing}
\label{sec:typing}

\begin{figure*}[!htb]
  \begin{small}
    \newcommand{\tc}[1]{{\color{red} #1}}
    \centering
    $\displaystyle
    \begin{array}{c}
    \infer[\id]
    {(a : A_m) \seq \tc{c \leftarrow a} :: (c : A_m)}
    {\mathstrut}
    \qquad
    \infer[\cut(S)]
    {\Psi \ \Psi' \seq \tc{S \leftarrow (\nu x)P; Q} :: (c : C_k)}
    {\Psi \geq m \geq k & |S| \sim m
        &\Psi \seq \tc{P} :: (x : A_m)
        &(S : A_m) \ \Psi' \seq \tc{Q} :: (c : C_k)
    }
    \\[1em]
    
    \infer[{\oplus} R_\ell^0]
    {(a : A_m^\ell) \seq \tc{c.\m{\ell}(a)} :: (c : \mathop{\oplus}\limits_{i \in I} A_m^i)}
    {\ell \in I}
    \qquad
    \infer[{\oplus} L]
    {\Psi \ (a : \mathop{\oplus}\limits_{i \in I} A_m^i) \seq \tc{\m{case}\; a\; (i(x_i) \Rightarrow P_i)_{i \in I}} :: (c : C_k)}
    {\Psi \ (x_i : A_m^i) \seq \tc{P_i} :: (c : C_k) \text{ for each $i \in I$}}
    \\[1em]
    
    \infer[{\with} R]
    {\Psi \seq \tc{\m{case}\; c\; (i(x_i) \Rightarrow P_i)_{i \in I}} :: (c : \mathop{\with}\limits_{i \in I} A_m^j)}
    {\Psi \seq \tc{P_i} :: (x_i : A_m^i) \text{ for each $i \in I$}}
    \qquad
    \infer[{\with} L_\ell^0]
    {(a :\mathop{\with}\limits_{i \in I} A_m^i) \seq \tc{a.\ell(c)} :: (c : A_m^\ell)}
    {\ell \in I}
    \\[1em]
    
    \infer[{\otimes} R^0]
    {(a : A_m) \ (b : B_m) \seq \tc{c.\langle a, b \rangle} :: (c : A_m \otimes B_m)}
    {\mathstrut}
    \qquad
    \infer[{\otimes} L]
    {\Psi \ (a : A_m \otimes B_m) \seq \tc{\m{case}\; a(\langle x, y \rangle \Rightarrow P)} :: (c : C_k)}
    {\Psi \ (x : A_m) \ (y : B_m) \seq \tc{P} :: (c : C_k)}
    \\[1em]
    
    \infer[{\one} R]
    {\cdot \seq \tc{c.\langle \rangle} :: (c : \one_m)}
    {\mathstrut}
    \qquad
    \infer[{\one} L]
    {\Psi \ (a : \one_m) \seq \tc{\m{case}\; a(\langle \rangle \Rightarrow P)} :: (c : C_k)}
    {\Psi \seq \tc{P} :: (c : C_k)}
    \\[1em]
    
    \infer[{\lolli} R]
    {\Psi \seq \tc{\m{case}\; c(\langle x, y \rangle \Rightarrow P)} :: (c : A_m \lolli B_m)}
    {(x : A_m) \ \Psi \seq \tc{P} :: (y : B_m)}
    \qquad
    \infer[{\lolli} L^0]
    {(a : A_m) \ (c : A_m \lolli B_m) \seq \tc{c.\langle a, b \rangle} :: (b : B_m)}
    {\mathstrut}
    \\[1em]
    
    \infer[{\upshift} R]
    {\Psi \seq \tc{\m{case}\; c(\m{shift}(x) \Rightarrow P)} :: (c : \upshift_k^m A_k)}
    {\Psi \seq \tc{P} :: (x :A_k)}
    \qquad
    \infer[{\upshift} L^0]
    {(a : \upshift_k^m A_k) \seq \tc{a.\m{shift}(c)} :: (c : A_k)}
    {\mathstrut}
    \\[1em]
    
    \infer[{\downshift} R^0]
    {(a : A_m) \seq \tc{c.\m{shift}(a)} :: (c : \downshift_k^m A_m)}
    {\mathstrut}
    \qquad
    \infer[{\downshift} L]
    {\Psi \ (a : \downshift_k^m A_m) \seq \tc{\m{case}\; a(\m{shift}(c) \Rightarrow P)} ::(c : C_\ell)}
    {\Psi \ (x : A_m) \seq \tc{P} :: (c : C_\ell)}
    \end{array}
    $
  \end{small}
    \caption{Process Assignment for Asynchronous Adjoint Logic}
    \label{fig:message-passing-typing-rules}
\end{figure*}

Our typing judgment for processes
$P$ is based on \emph{intuitionistic sequents} of the form
\[
  (x^1 : A^1) \cdots (x^n : A^n) \vdash P :: (x : A)
\]
where each of the $x^i$ are \emph{channels} that
$P$ \emph{uses} and $x$ is a channel that
$P$ \emph{provides}.  All of these channels must be distinct and we
abbreviate the collection of antecedents as
$\Psi$.  The \emph{session types} $A^i$ and
$A$ specify the communication behavior that the process
$P$ must follow along each of the channels.

Such sequents are standard for the intuitionistic approach to
understanding binary session types (e.g.,~\cite{Caires10concur}) where
the channels are \emph{linear} in that every channel in a network of
processes has exactly one provider and exactly one client.  In the
closely related formulation based on classical linear
logic~\cite{Wadler12icfp} all channels are on the right-hand side of
the turnstile, but each linear channel still has exactly two
endpoints.

We generalize this significantly by assigning to each channel an
intrinsic \emph{mode} $m$.  Each mode $m$ is assigned a set of
structural properties $\sigma(m)$ among W (for weakening) and C (for
contraction). Separating $m$ from $\sigma(m)$ allows us to have multiple
modes with the same set of structural properties.\footnote{
This allows us, for example, to model the modal logic S4 or lax logic
(the logical origins of comonadic and monadic programming),
each with two modes both satisfying weakening and contraction,
as well as linear analogues of these constructions.}
No matter which structural properties are available for
a channel, each active channel will still have \emph{exactly one
  provider}.  Beyond that, a channel $x_m$ with ${\rm W} \in \sigma(m)$ may
not have any clients.  Furthermore, a channel $x_m$ with
${\rm C} \in \sigma(m)$ may have multiple clients.  All other properties of
our system of session types for processes derive systematically from
these simple principles.

The modes are organized into a preorder where $m \geq k$ requires that
$\sigma(m) \supseteq \sigma(k)$, that is, $m$ must allow more
structural properties than $k$.  In order to guarantee session
fidelity and deadlock freedom, for any sequent
$\Psi \vdash P :: (x_m : A_m)$ is must be the case that for every
$y_k : B_k \in \Psi$ we have $k \geq m$.  For example, if $m$ permits
contraction and therefore $P$ may have multiple clients, then for any
$y_k$ in $\Psi$, mode $k$ must also permit contraction because
(intuitively) if $x_m$ is referenced multiple times then, indirectly,
so is $y_k$. If $k \geq m$ then this is ensured. We express this with
the \emph{presupposition} that
\[
  \Psi \vdash P :: (x_m : A_m)\quad \mbox{requires} \quad \Psi \geq m
\]
where $\Psi \geq m$ simply means $k \geq m$ for every
$y_k : A_k \in \Psi$.  We will only consider sequents satisfying this
presupposition, so our rules, when they are used to break down a
conclusion into the premises, must preserve this fundamental property
which we call \emph{the declaration of independence}.

In our formulation, channels $x_m$ as well as types $A_m$ are endowed
with modes which must always be consistent between a channel and its
type ($x_m : A_m$).  We therefore often omit redundant mode
annotations on channels.

The complete set of rules for the typing judgment are given in
Fig.~\ref{fig:message-passing-typing-rules}.
We first examine the judgmental rules that explain the meaning of
identity and composition.  Identity (rule $\id$) is straightforward:
a process $c \leftarrow a$ providing $c$ defers to the provider of
$a$, which is possible as long as $a$ and $c$ have the same
type and mode.  This is usually called \emph{forwarding}
or \emph{identification} of the channels $a$ and $c$.

% Identity is straightforward (rule
% $\id$ in Fig.~\ref{fig:message-passing-typing-rules})
% \[
%   \infer[\id]
%     {(a : A_m) \vdash c \leftarrow a :: (c : A_m)}
%     {}
% \]
% the provider of $c$ can defer to the provider of $a$ that it is a
% client of, as long as $a$ and $c$ have the same type (which implies
% that they also have the same mode).  We call this \emph{forwarding}
% and sometimes we say that we \emph{identify} the channels $a$ and $c$.

The usual logical rule of cut corresponds to the parallel composition of two
processes with a single private channel for communication between
them.
% \[
%   \infer[\cut]
%     {\Psi \ \Psi' \seq x \leftarrow P \semi Q :: (c : C_k)}
%     {\Psi \ge m \ge k
%     &\Psi \seq P :: (x : A_m)
%     &\Psi' \ (x : A_m) \seq Q :: (c : C_k)
%     }
% \]
% Here, $P$ is the provider of a channel $x$ and $Q$ its client.  
However, ordinary cut is insufficiently general to describe the
situation where a single provider of a channel $x_m$ may have multiple
clients ($C \in \sigma(m)$) or no clients ($W \in \sigma(m)$).  We
therefore generalize it to a form of multicut,
\footnote{The term "multicut" has been used in the literature for different rules.
We follow here the proof theory literature~\cite[Section 5.1]{Negri01book},
where it refers to a rule that cuts out some number of copies of the \emph{same}
proposition A, as in Gentzen's original proof of cut elimination~\cite{Gentzen35},
where he calls it ``Mischung''.}
 where the channel $x_m$
provided by $P$ is known by multiple aliases in the set of channels
$S$ in $Q$ as long as the multiplicity of the aliases is permitted by
the mode.  This is expressed as $|S| \sim m$ and is sufficient for
static typing. Formally, we define this condition by $0 \sim m$ if
$\m{W} \in \sigma(m)$, $1 \sim m$ always, and $k \sim m$ for
$k \geq 2$ if $\m{C} \in \sigma(m)$.  When processes execute we will
have an even more general situation where one provider has multiple
separate client processes, which is captured in the typing judgment
for process configurations (\cref{sec:operational}).

Next we come to the various session types.  From the logical
perspective, these are the propositions of adjoint logic.
\[
\begin{array}{lcl}
A_m, B_m & \Coloneqq & p_m \mid A_m \lolli_m B_m \mid A_m \tensor_m B_m \mid \one_m \mid \mathop{\oplus_m}\limits_{i \in I} A_m^i
\mid \mathop{\with_m}\limits_{i \in I} A_m^i \mid \up_k^m A_k \mid \down_m^\ell A_\ell
\end{array}
\]
Here, $p_m$ stands for atomic propositions at mode $m$.  The other connectives, other than $\up_k^m$
and $\down_m^\ell$, are standard linear logic connectives, except that they are only allowed to
combine types (propositions) at the same mode. Since the mode of a connective can be inferred from the
modes of the types it connects (other than for shifts), we omit subscripts on connectives. Note also that
$\with$ and $\oplus$ have been generalized to $n$-ary forms from the usual binary forms. This is
convenient for programming. We will use a label set $I = \{\pi_1, \pi_2\}$ when working with the
binary forms $A_m \with B_m$ and $A_m \oplus B_m$, where $\pi_1$ selects the left-hand type and
$\pi_2$ selects the right-hand type.
The operational meaning of these connectives (as
discussed further in \cref{sec:operational}) is largely similar to that in past work
(e.g. \cite{Caires10concur}), with $\lolli_m$ and $\tensor_m$ sending channels along other channels,
$\one_m$ sending an end-of-communication message, and $\oplus_m$ and $\with_m$ sending labels. The
shifts send a simple $\m{shift}$ message to signal a transition between modes, either \emph{up} ($\up_k^m$)
from $k$ to some $m \geq k$ or \emph{down} ($\down_m^\ell$) from $\ell$ to some $m \leq \ell$.

We provide proof terms for the rules in our sequent calculus, as shown in
Figure~\ref{fig:message-passing-typing-rules}.  We can then interpret the proof terms as process
expressions, and these rules are used to give the typing judgment for such processes.
\Cref{tbl:proc-meanings} gives the informal meaning of each such process term.

In general, our process syntax represents an intermediate point between a programmer-friendly syntax
and a notation in which it is easy to describe the operational semantics and prove progress and
preservation.  When compared to, for instance, SILL~\cite{Toninho13esop}, the main revisions are
that (1)~we make channel continuations explicit in order to facilitate asynchronous communication
while preserving message order~\cite{DeYoung12csl}, and (2)~we distinguish between an \emph{internal
  name} for the channel provided by a process and \emph{external names} connecting it to multiple
clients.

\begin{table}
    \centering
    \begin{tabular}{lp{4.5in}}
        \toprule
        Process term & Meaning\\
        \midrule
        $a \leftarrow c$  & Identify channels $a$ and $c$.\\
        $S \leftarrow (\nu x)P  \semi Q$ & Spawn a new process $P$ providing channel $x$ with aliases $S$ to be used by $Q$. Here, $x$ is the \emph{internal name} in $P$ for the channel offered by $P$, and $S$ is the set of \emph{external names} of the same channel as used in $Q$. \\[1.2ex]
        $c.\ell(a)$ & Send the label $\ell$ and the channel $a$ along $c$.\\
        $\m{case}\; c (i(x_i) \Rightarrow P_i\}_{i \in I})$ & Receive a label $i$ and a channel $x_i$ from
        $c$, continue as $P_i$.\\[1.2ex]
        $c.\langle a, b \rangle $ & Send the channels $a$ and $b$ along $c$.\\
        $\m{case}\; c(\langle x, y \rangle \Rightarrow P)$ & Receive channels $x$ and $y$ from $c$ to be used in $P$.\\[1.2ex]
        $c.\langle \rangle$ & End communication over $c$ by sending a terminal message.\\
        $\m{case}\; c(\langle \rangle \Rightarrow P)$ & Wait for $c$ to be closed, continue as $P$.\\[1.2ex]
        $c_m.\m{shift}(a_k)$ & Send a shift, from mode $m$ to mode $k$ \\
        $\m{case}\; c_m\; (\m{shift}(x_k) \Rightarrow P)$ & Receive a shift from mode $m$ to mode $k$ \\[0.2ex]
        \bottomrule
    \end{tabular}
    \caption{Informal Meanings of Process Terms}
    \label{tbl:proc-meanings}
\end{table}

\pgraph{Some simple examples.}
% Now we can show some actual processes.
We provide here some small examples with their types;
additional examples which highlight more
interesting behavior can be found in \cref{sec:example}.

First, we have a process that can be written at any mode $m$, which witnesses
that $\otimes_m$ is commutative.
\begin{tabbing}
    $(x : A_m \otimes B_m) \vdash \m{case}\; x\,(\langle y, x' \rangle \Rightarrow z.\langle x', y \rangle) :: (z : B_m \otimes A_m)$
\end{tabbing}

If $m$ is a mode that admits contraction, we can write the following process,
which witnesses that $A_m \with B_m$ proves $A_m \otimes B_m$ in the
presence of contraction.  `$\%$' starts a comment.
\begin{tabbing}
    $(p : A_m \with B_m) \vdash$
    \= $\{p_1, p_2\} \leftarrow (\nu q)\, (q \leftarrow p); \qquad$
    \= $\%\ \{p_1, p_2\} \leftarrow \m{copy}\; p$\\
    \> $x \leftarrow (\nu a)\, p_1.\pi_1(a);$\\
    \> $y \leftarrow (\nu b)\, p_2.\pi_2(b);$\\
    \> $z.\langle x, y \rangle$ \> $:: (z : A_m \otimes B_m)$
\end{tabbing}

If $m$ is a mode that admits weakening, we can write the following process,
which witnesses that $A_m \otimes B_m$ proves $A_m \with B_m$ in the
presence of weakening.
\begin{tabbing}
    $(x : A \otimes B) \vdash$
    \= $\m{case}\; p\;$
    \= $($\,%
    \= $\pi_1(p_1) \Rightarrow$
    \= $\m{case}\; x\,(\langle y, z \rangle \Rightarrow$\\
    \>\>\>\> $\{\,\} \leftarrow (\nu a)\, (a \leftarrow z); \qquad \quad$
    \= $\%\ \m{drop}\; z$\\
    \>\>\>\> $p_1 \leftarrow y)$\\
    \>\> $|\,\pi_2(p_2) \Rightarrow$
    \= $\m{case}\; x\,(\langle y, z \rangle \Rightarrow$\\
    \>\>\>\> $\{\,\} \leftarrow (\nu a)\, (a \leftarrow y); \qquad \quad$
    \= $\%\ \m{drop}\; y$\\
    \>\>\>\> $p_2 \leftarrow z \,))$\\
    $:: (p : A \with B)$
\end{tabbing}

\section{Operational Semantics}
\label{sec:operational}

\newcommand{\proc}{\m{proc}}

In order to describe the computational behavior of process
expressions, we need to first give some syntax for the computational
artifacts, which are running processes $\proc(S, \Delta, a, P)$.  Such
a process executes $P$ and provides a channel $a$ while using the channels
in the channel set $\Delta$. $S$ is a set of aliases for the channel $a$,
which can be referred to by one or more clients. Each alias $c \in S$ is used
by at most one client, but one client may use multiple such aliases. Note that
as the aliases in $S$ are the only way to interact with the channel $a$ from
an external process, the objects $\proc(S, \Delta, a, P)$ and
$\proc(S, \Delta, b, P[b/a])$ are equivalent --- changing the internal name of
a process has no effect on its interactions with other processes.

A \emph{process configuration} is a multiset of processes:
\[
\begin{array}{lcl}
\CC & \Coloneqq & \proc(S, \Delta, a, P) \mid (\cdot) \mid \CC\ \CC'
\end{array}
\]
where we require that all the aliases or names provided by the processes
$\proc(S, \Delta, a, P)$ are distinct, i.e., given objects
$\proc(S, \Delta_1, a, P)$ and $\proc(T, \Delta_2, b, Q)$ in the same process
configuration, $S$ and $T$ are disjoint.  We will specify the
operational semantics in the form of
\emph{multiset rewriting rules}~\cite{Cervesato09ic}.
That means we show how to rewrite some subset of the configuration
while leaving the remainder untouched. This form provides some
assurance of the locality of the rules.

It simplifies the description of the operational semantics if for any
process $\proc(S, \Delta, a, P)$, $\Delta$ consists of exactly the free
channels (other than $a$) in $P$. This requires that we restrict the labeled internal
and external choices, $\mathop{\oplus}\limits_{i \in I} A^i_m$ and
$\mathop{\with}\limits_{i \in I} A^i_m$ to the case where
$I \neq \emptyset$. Since a channel of empty choice type can never
carry any messages, this is not a significant restriction in practice.

% \subsubsection{Configuration typing}
% \label{sssec:configuration-typing}

In order to understand the rules of the operational semantics, it will
be helpful to understand the typing of configurations.  The judgment
has the form $\Psi \vDash \CC :: \Psi'$ which expresses that using the
channels in $\Psi$, configuration $\CC$ provides the channels in
$\Psi'$. This allows a channel that is not mentioned at all in $\CC$
to appear in both $\Psi$ and $\Psi'$---we think of such a channel as
being ``passed through'' the configuration.

Note that while the configuration typing rules induce an ordering on a
configuration, the configuration itself is not inherently ordered.
The key rule is the first: for any object $\proc(S, \Delta, a, P)$ we
require that $P$ is well-typed on some subset of the
available channels while the others are passed through.  Here we write
$\overline{\Psi}$ for the set of channels declared in $\Psi$, which
must be exactly those used in the typing of $P$.  Moreover,
\emph{externally} such a process provides the channels
$S = \{a^1_m, \ldots, a^n_m\}$, all of the same type $A_m$.  We use
the abbreviation $(S : A_m)$ for $a^1_m : A_m, \ldots, a^n_m : A_m$.
Finally, we enforce that the number of clients must be compatible with
the mode $m$ of the offered channel, which is exactly that
$|S| \sim m$, as defined in \cref{sec:typing}.
\[
\infer[\m{Proc}]
{\Psi\ \Psi' \vDash \proc(S, \overline{\Psi'}, a, P) :: \Psi\ (S : A_m)}
{|S| \sim m
    &\Psi' \seq P :: (a : A_m)
}
\qquad
\infer[\m{Id}]
{\Psi \vDash (\cdot) :: \Psi}
{\mathstrut}
\qquad
\infer[\m{Comp}]
{\Psi \vDash \mathcal C\ \mathcal C' :: \Psi''}
{\Psi \vDash \mathcal C :: \Psi'
    &\Psi' \vDash \mathcal C' :: \Psi''
}
\]
The identity and composition rules are straightforward.
The empty context $(\cdot)$ provides $\Psi$ if given $\Psi$, since it does not
use any channels in $\Psi$ or provide any additional channels.
Composition just connects configurations with compatible interfaces:
what is provided by $\CC$ is used by $\CC'$.

% \subsection{Computation rules}

\begin{figure*}[!tbh]
    \centering
    $\displaystyle
    \begin{array}{rcl}
    \begin{array}{r}
    \proc(T \cup \{c\}, \Delta, x, P)\\
    \proc(S, \{c\}, y, y \leftarrow c)\\
    \end{array}
    &\overset{\id}{\Longrightarrow}&
    \begin{array}{l}
    \proc(T \cup S, \Delta, x, P)\\
    \end{array}
    \\[1.2em]
    
    \begin{array}{r}
    \proc(T, \Delta_P \cup \Delta_Q, y, S \leftarrow (\nu x)P; Q)\\
    (\text{$S'$ a fresh set of channels matching $S$})
    \end{array}
    &\overset{\cut(S)}{\Longrightarrow}&
    \begin{array}{l}
    \proc(S', \Delta_P, x, P) \\
    \proc(T, \Delta_Q \cup \{S'\}, y, Q[S'/S])\\
    \end{array}
    \\[1.2em]
    
    \begin{array}{r}
    (P \text{ not an identity}) \qquad \proc(\emptyset, \Delta, x, P)
    \end{array}
    &\overset{\m{drop}}{\Longrightarrow}&
    \begin{array}{l}
    \proc(\emptyset, \{b\}, y, y \leftarrow b)_{b \in \Delta} \\
    \end{array}
    \\[1.2em]
    
    \begin{array}{r}
    \proc(S \cup T, \Delta, x, P)\\
    (\text{$P$ not an identity and $S, T$ non-empty})
    \end{array}
    &\overset{\m{copy}}{\Longrightarrow}&
    \begin{array}{l}
    \proc(\{b', b''\}, \{b\}, y, y \leftarrow b)_{b \in \Delta}\\
    \proc(S, \{b'\}_{b \in \Delta}, x, P[b'/b])\\
    \proc(T, \{b''\}_{b \in \Delta}, x, P[b''/b])\\
    \end{array}
    \\[1.2em]
    
    \begin{array}{r}
    \proc(\{b\}, \{c\}, x, x.\ell(c))\\
    \proc(S, \Delta \cup \{b\}, z, \m{case}\; b(i(y_i) \Rightarrow P_i)_{i \in I})\\
    \end{array}
    &\overset{{\oplus}\; C}{\Longrightarrow}&
    \begin{array}{l}
    \proc(S, \Delta \cup \{c\}, z, P_\ell[c/y_\ell])\\
    \end{array}
    \\[1.2em]
    
    \begin{array}{r}
    \proc(\{b\}, \Delta, x, \case\; x(i(y_i) \Rightarrow P_i)_{i \in I})\\
    \proc(\{c\}, \{b\}, z, b.\ell(z))\\
    \end{array}
    &\overset{{\with}\; C}{\Longrightarrow}&
    \begin{array}{l}
    \proc(\{c\}, \Delta, z, P_\ell[z/y_\ell])\\
    \end{array}
    \\[1.2em]
    
    \begin{array}{r}
    \proc(\{b\}, \{c, d\}, w, w.\langle c, d \rangle)\\
    \proc(S, \Delta \cup \{b\}, z, \m{case}\; b(\langle x, y \rangle \Rightarrow P)\\
    \end{array}
    &\overset{{\otimes}\; C}{\Longrightarrow}&
    \begin{array}{l}
    \proc(S, \Delta \cup \{c, d\}, z, P[c/x, d/y])\\
    \end{array}
    \\[1.2em]
    
    \begin{array}{r}
    \proc(\{b\}, \Delta, w, \m{case}\; w(\langle x, y \rangle \Rightarrow P)\\
    \proc(\{c\}, \{b, d\}, z, b.\langle d, z \rangle)\\
    \end{array}
    &\overset{{\lolli}\; C}{\Longrightarrow}&
    \begin{array}{l}
    \proc(\{c\}, \Delta \cup \{d\}, z, P[d/x, z/y])\\
    \end{array}
    \\[1.2em]
    
    \begin{array}{r}
    \proc(\{b\}, \emptyset, x, x.\langle \rangle)\\
    \proc(S, \Delta \cup \{b\}, y, \m{case}\; b (\langle \rangle \Rightarrow P))\\
    \end{array}
    &\overset{{\one}\; C}{\Longrightarrow}&
    \begin{array}{l}
    \proc(S, \Delta, y, P)\\
    \end{array}
    \\[1.2em]
    
    \begin{array}{r}
    \proc(\{b_k\}, \{c_m\}, x_k, x_k.\m{shift}(c_m))\\
    \proc(S, \Delta \cup \{b_k\}, y, \m{case}\; b_k(\m{shift}(z_m) \Rightarrow P))\\
    \end{array}
    &\overset{{\down_k^m}\; C}{\Longrightarrow}&
    \begin{array}{l}
    \proc(S, \Delta \cup \{c_m\}, y, P[c_m/z_m])\\
    \end{array}
    \\[1.2em]
    
    \begin{array}{r}
    \proc(\{b_m\}, \Delta, x_m, \case\; x_m(\m{shift}(z_k) \Rightarrow P))\\
    \proc(\{c_k\}, \{b_m\}, y_k, b_m.\m{shift}(y_k))\\
    \end{array}
    &\overset{{\up_k^m}\; C}{\Longrightarrow}&
    \begin{array}{l}
    \proc(\{c_k\}, \Delta, y_k, P[y_k/z_k])\\
    \end{array}
    \\[1.2em]
    \end{array}
    $
    \caption{Computation Rules for Asynchronous Adjoint Logic}
    \label{fig:dynamics1}
\end{figure*}

The computation rules we discuss in this section can be found in
Figure~\ref{fig:dynamics1}.
Remarkably, the computation rules do not depend on the modes, although
some of the rules will naturally only apply at modes satisfying
certain structural properties.

\pgraph{Judgmental rules.}
The identity rule (written as $\overset{\id}{\Longrightarrow}$)
describes how an identity process (for instance,
${\proc(S, \{c\}, a, a \leftarrow c)}$) may interact with other
processes.  We think of such a process as connecting the provider of
$c$ to clients in $S$, and therefore sometimes call it a
\emph{forwarding process}.
A forwarding process interacts with
the provider of $c$, telling it to replace $c$ with $S$ in its set of
clients. In adding $S$ to the set of clients, the forwarding process
accomplishes its goal of connecting the provider of $c$ to $S$, and so
it can terminate.

The cut rule steps by spawning a new process which offers along a
fresh set of channels $S'$, all of which are used in $Q$, the
continuation of the original process.  Here we write $\Delta_P$ and
$\Delta_Q$ for the set of free channels in $P$ and $Q$, respectively.

\pgraph{Structural rules.}
A process with no clients can terminate (rule
$\overset{\m{drop}}{\Longrightarrow}$), but must notify all of the
processes it uses that they should also terminate.  It does so by
sending each one a forwarding message, effectively embodying a
cancellation.  In concert with the identity rule this accomplishes
cascading cancellation in the distributed setting.  Note that the mode
$m$ of channel $a$ must admit weakening in order for the process on
the left-hand side of the rule to be well-typed.

Similarly, a process with multiple clients can spawn a copy of itself,
each with a strictly smaller set of clients (rule
$\overset{\m{copy}}{\Longrightarrow}$).  If the process $P$ is a
replicable service, that is, if it has a negative type $\with$,
$\lolli$, ${\up}^m_k$, then this corresponds to actual process
replication.  If it has a positive type $\oplus$, $\otimes$, $\one$,
${\down}^m_k$, this corresponds to duplicating a multicast message
into copies for different subsets of recipients.  The mode
$m$ of the channel $a$ must admit contraction in order for the process
on the left-hand side of the rule to be well-typed.

While both the $\m{drop}$ and $\m{copy}$ rules can be applied to any process
with $0$ or multiple clients, respectively, this does not cause any problems
as long as we forbid them from executing on identity processes. If we apply drop
or copy to an identity process, we end up with another process of the same form
on the right-hand side of the rule, and so we could repeatedly apply drop or
copy and not make any progress. As such, we forbid this use of the drop and copy
rules.

For any other type of process, regardless of whether we drop/copy first or execute
another communication rule first, we can eventually reach the same state, and so
we do not need to make additional restrictions (though an actual implementation would
likely pick either a maximally eager or a maximally lazy strategy for applying
these rules). 

\pgraph{Additive and multiplicative connectives.}
In the rule for $\oplus$, the process
$\m{proc}(\{b\}, \{c\}, a, a.\ell(c))$ represents the
message `label~$\ell$ with continuation~$c$'.  After this message has
been received, the process terminates since $b$ was its only client.
The recipient selects the appropriate branch of the $\m{case}$
construct and also substitutes the continuation channel $c$ for the
continuation variable $d_\ell$.

The $\with$ computation rule is largely similar to that for
$\oplus$, except that communication proceeds in the opposite
direction---messages are sent \emph{to} providers \emph{from} clients, rather
than from providers to clients as in the case of $\oplus$.

The multiplicative connectives $\otimes$ and $\lolli$ behave similarly
to their additive counterparts, except that rather than sending and
receiving labels, they send and receive channels together with a
continuation, and so an extra substitution is required when receiving
messages, while the $\one$ behaves as a nullary $\otimes$, allowing us to
signal that no more communication is forthcoming along a channel, and
to wait for such a signal before continuing to compute.

\pgraph{Shifts.}
We present the computation rules for shifts with modes marked explicitly
on the relevant channels. Channels whose modes are unmarked may be at any
mode (provided, of course, that the declaration of independence is respected).

Operationally, $\upshift$ behaves essentially the same as unary
$\with$, while $\downshift$ behaves as unary $\oplus$.
Their significance lies in the \emph{mode shift} of the continuation
channel that is transmitted, which is required for the configuration
to remain well-typed.

The messages $\m{shift}(a_k)$ or $\m{shift}(c_m)$ should be thought of as
signaling a transition between modes --- to mode $k$ for the former, and to mode
$m$ for the latter. Whether the transition is up or down depends on which direction
the message is being sent in. As with other messages (in particular, the messages
for $\oplus$ and $\with$), the continuation channels are made explicit.

\section{Session Fidelity, Deadlock-Freedom, and Garbage Collection}
\label{sec:metatheory}

While we can prove cut elimination for the form of adjoint logic presented
in \cref{app:adjoint-logic}, from a programmer's perspective we
are not interested in eliminating all cuts (which would
correspond to reducing under $\lambda$-abstractions in a functional
language) but rather we block when waiting to receive a message, analogous to
a $\lambda$-abstraction waiting for input before it can reduce.
What we prove instead are session fidelity and deadlock-freedom.

\pgraph{Session fidelity.}
The session fidelity theorem follows from a case analysis on the computation
rule used to get that $\mathcal C \Rightarrow \mathcal C'$. In each case, we
break $\mathcal C$ down to find the processes on which the computation rule acts,
along with some collections of processes which are unaffected by the computation.
From these pieces, we build a proof that $\Psi \vDash \mathcal C' :: \Psi'$.

\begin{theorem}[Session Fidelity]
    If $\Psi \vDash \mathcal C :: \Psi'$ and $\mathcal C \Rightarrow \mathcal C'$, then $\Psi \vDash \mathcal C' :: \Psi'$. 
\end{theorem}

\pgraph{Deadlock-freedom.}
The progress theorem for a functional language states that an expression
is either a value or it can take a step.  Here we do not have values, but
there is nevertheless a clear analogue between, say, a value $\lambda x.e$
that waits for an argument, and a process $\m{case}\; x\, (\langle y, z \rangle \Rightarrow P)$
that waits for an input. We formalize this in the definition below.

\begin{definition}
    We say that a process $\proc(S, \Delta, a, P)$ is \emph{poised} on $a$ if:
    \begin{enumerate}
        
        \item it is a process $\proc(S, \Delta, a, P)$ that sends on $a$ --- that is, $P$ is of
        the form $(a.\uscore)$, or
        \item it is a process $\proc(S, \Delta, a, P)$ that receives on $a$ --- that is, $P$ is
        of the form $(\m{case}\; a\; (\uscore))$.
    \end{enumerate}
\end{definition}

Intuitively, $\proc(S, \Delta, a, P)$ is poised on $a$ if it is blocked
trying to communicate along $a$. This definition allows us to state the
following progress theorem:

\begin{theorem}[Deadlock-Freedom]
    If $(\cdot) \vDash \mathcal C :: \Psi$, then exactly one of the following holds:
    \begin{enumerate}
        \item There is a $\mathcal C'$ such that $\mathcal C \Rightarrow \mathcal C'$.
        \item Every $\proc(S, \Delta, a, P)$ in $\mathcal C$ is poised on $a$.
    \end{enumerate}
\end{theorem}

In order to prove this theorem, we first prove a lemma
allowing us to take advantage of the ordering induced
by configuration typing. We note that if object $\psi$
is a client of object $\phi$, $\psi$ must occur to the right
of $\phi$ in the ordering, and so if we can analyze a
configuration from right to left, we consider each process
before (or after, depending on your view of induction) all
of its dependencies. To formalize this,
we present a second set of rules defining another form
of configuration typing (which will turn out to prove the same
judgments as the original form).
\[
\infer[\m{Empty}]
{\Psi \vDash' (\cdot) :: \Psi}
{\mathstrut}
\qquad
\infer[\m{Extend}]
{\Psi \vDash' \mathcal C\  \proc(S, \overline{\Psi'}, a, P) :: \Psi''\ (S : A_m)}
{|S| \sim m
    &\Psi \vDash' \mathcal C :: \Psi'\ \Psi''
    &\Psi' \seq P :: (a : A_m)
}
\]
It is clear that if $\vDash$ and $\vDash'$ are the same, then
we can perform induction using the $\m{Empty}$ and $\m{Extend}$ rules
rather than the $\m{Id}$, $\m{Comp}$, and $\m{Proc}$ rules, allowing us
to analyze a configuration from right to left. We formalize this as
\cref{lem:config-order}.

\begin{lemma}
\label{lem:config-order}
    $\Psi \vDash \mathcal C :: \Psi'$ if and only if $\Psi \vDash' \mathcal C :: \Psi'$.
\end{lemma}

This lemma is nearly immediate --- all of the rules for $\vDash'$ are derivable from the rules
of $\vDash$, and all rules of $\vDash$ but $\m{Comp}$ are derivable from the rules of $\vDash'$.
We therefore need only show (by an induction over the right-hand premise) that the version of the
$\m{Comp}$ rule with $\vDash$ replaced by $\vDash'$ is admissible.

The proof of deadlock-freedom then proceeds by an induction on the derivation of
$(\cdot) \vDash \mathcal C :: \Psi$, using \cref{lem:config-order} to work
right to left. Writing $\mathcal C = \mathcal C'\; \proc(S, \overline{\Psi'}, a.P)$,
we see that either $\mathcal C'$ can step, in which case so can
$\mathcal C$, or every process in $\mathcal C'$ is poised.
Now we carefully distinguish cases on $S$ (empty, singleton, or
greater) and apply inversion to the typing of $P$ to see that in
each case the process either is poised, can take a step independently,
or can interact with provider of a channel in $\overline{\Psi'}$.

\pgraph{Garbage collection.}
% \label{par:gc}
As we can see from the preservation theorem, the interface to a
configuration never changes.  While new processes may be spawned, they
will have clients and are therefore not visible at the interface.
This is in contrast to the semantics of shared channels in prior work
(for example, in~\cite{Caires10concur,Pfenning15fossacs}) where shared
channels may show up as newly provided channels.  Therefore they may
be left over at the end of a computation without any clients.

This cannot happen here.  Initially, at the top level, we envision
starting with the configuration below on the left. Assuming this computation
completes, by the progress property and the definition of \emph{poised}, computation
could only halt with the configuration on the right. In other words: no garbage!
\[ \cdot \vDash \proc(\{c_0\}, \cdot, c, P_0) :: (c_0 : \one) \qquad \qquad
\cdot \vDash \proc(\{c_0\}, \cdot, c, c.\langle \rangle) :: (c_0 : \one)\]

One can generalize this to allow nontrivial output by
allowing any purely positive type (that is, one which only uses the fragment
of the logic with connectives $\oplus$, $\otimes$, $\one$, and $\downshift$),
such as $\oplus\{\m{false} : \one, \m{true} : \one\}$.

We can formalize this intuition by defining an \emph{observable} configuration $\mathcal C$
which corresponds to our intuitive notion of garbage-free. We only define what it means for
a configuration with purely positive type to be observable. It is likely that this definition
can be extended to encompass negative types as well, but it is not nearly as natural to do so.

A configuration $\mathcal C$ for which there is $\Psi$ composed
entirely of purely positive types such that $\cdot \vDash \mathcal C :: \Psi$
is \emph{observable} at $\Psi$ if,
when we repeatedly receive messages from all channels we know about, starting from a state
where we only know about $\Psi$, we eventually receive a message from every object in $\mathcal C$.
If we do not care about the particular channels in $\Psi$, we may say simply that $\mathcal C$ is
\emph{observable}.

\begin{definition}
We define what it means for a configuration $\mathcal C$ to be observable at $\Psi$
(written $\mathcal C \rhd \Psi$) inductively over the structure of $\mathcal C$.
\begin{enumerate}
\item $\proc(\{c\}, \cdot, x, x.\langle \rangle) \rhd (c : \one)$.
\item If $\mathcal C \rhd \Psi\; (d : A_m^\ell)$, then
$\mathcal C\; \proc(\{c\}, \{d\}, x, x.\ell(d)) \rhd \Psi\; (c : \mathop{\oplus}\limits_{i \in I} A_m^i)$.
\item If $\mathcal C \rhd \Psi\; (d : A_m)$, then
$\mathcal C\; \proc(\{c\}, \{d\}, x, x.\m{shift}(d)) \rhd \Psi\; (c : \down_k^m A_m)$.
\item If $\mathcal C \rhd \Psi\; (d : A_m)\; (e : B_m)$, then
$\mathcal C\; \proc(\{c\}, \{d, e\}, x, x.\langle d, e \rangle) \rhd \Psi\; (c : A_m \otimes B_m)$.
\end{enumerate}
\end{definition}

We can then give the following corollary of our deadlock-freedom theorem:
\begin{corollary}
If $\cdot \vDash \mathcal C :: \Psi$ for some $\Psi$ consisting entirely of purely positive types and $\mathcal C$ cannot
take any steps, then $\mathcal C \rhd \Psi$.
\end{corollary}

This proof proceeds by a simple induction on the derivation of $\cdot \vDash \mathcal C :: \Psi$,
using (\Cref{lem:config-order}) to work from right to left. At each step, we note that the rightmost process is poised. Because
$\Psi$ consists only of purely positive types, the rightmost process must therefore be sending a positive message. Moreover,
it can only use channels of purely positive type. Well-typedness of the configuration then lets us apply the inductive hypothesis
to the remainder of the configuration, at which point we can simply apply the definition of observability.

\section{Conclusion}

At this point, our formulation of adjoint logic and its operational
semantics seem to provide a good explanation for a variety of
patterns of asynchronous communication. The key behaviors which we can
model (and importantly, model in a uniform fashion) are cancellation,
replication, and multicast. We also obtain a foundation for a system of distributed
garbage collection.  Moreover, if used linearly, our semantics coincides
with the purely linear semantics developed in prior work.

In parallel work we have also provided a shared memory semantics for a
closely related formulation of adjoint logic with implicit structural
rules~\cite{Pfenning18tlla}.  In future work, we plan to investigate
if the declaration of independence is sufficient to allow a
\emph{modular} combination of different operational interpretations
for different modes.  Of particular interest here would be the
semantics with manifest sharing~\cite{Balzer17icfp}.

%  We also have not yet explored the full
% range of examples suggested by instances of the adjoint logic
% framework, such as potential concurrent programming application of
% judgmental S4 (comonads) or lax logic (strong monads).

\section*{Acknowledgments}

Supported by NSF Grant No. CCF-1718267: ``Enriching Session Types for Practical Concurrent Programming''

\bibliographystyle{eptcs}
\bibliography{fp}

\begin{thebibliography}{10}
\providecommand{\bibitemdeclare}[2]{}
\providecommand{\surnamestart}{}
\providecommand{\surnameend}{}
\providecommand{\urlprefix}{Available at }
\providecommand{\url}[1]{\texttt{#1}}
\providecommand{\href}[2]{\texttt{#2}}
\providecommand{\urlalt}[2]{\href{#1}{#2}}
\providecommand{\doi}[1]{doi:\urlalt{http://dx.doi.org/#1}{#1}}
\providecommand{\bibinfo}[2]{#2}

\bibitemdeclare{inproceedings}{Balzer17icfp}
\bibitem{Balzer17icfp}
\bibinfo{author}{Stephanie \surnamestart Balzer\surnameend} \&
  \bibinfo{author}{Frank \surnamestart Pfenning\surnameend}
  (\bibinfo{year}{2017}): \emph{\bibinfo{title}{Manifest Sharing with Session
  Types}}.
\newblock In: {\sl \bibinfo{booktitle}{International Conference on Functional
  Programming (ICFP)}}, \bibinfo{publisher}{ACM}, pp.
  \bibinfo{pages}{37:1--37:29}, \doi{10.1145/3110281}.

\bibitemdeclare{inproceedings}{Benton94csl}
\bibitem{Benton94csl}
\bibinfo{author}{Nick \surnamestart Benton\surnameend} (\bibinfo{year}{1994}):
  \emph{\bibinfo{title}{A Mixed Linear and Non-Linear Logic: Proofs, Terms and
  Models}}.
\newblock In \bibinfo{editor}{Leszek \surnamestart Pacholski\surnameend} \&
  \bibinfo{editor}{Jerzy \surnamestart Tiuryn\surnameend}, editors: {\sl
  \bibinfo{booktitle}{Selected Papers from the 8th International Workshop on
  Computer Science Logic (CLS'94)}}, \bibinfo{publisher}{Springer LNCS 933},
  \bibinfo{address}{Kazimierz, Poland}, pp. \bibinfo{pages}{121--135},
  \doi{10.1007/BFb0022251}.
\newblock \bibinfo{note}{An extended version appears as Technical Report
  UCAM-CL-TR-352, University of Cambridge}.

\bibitemdeclare{inproceedings}{Caires2017linearity}
\bibitem{Caires2017linearity}
\bibinfo{author}{Lu{\'\i}s \surnamestart Caires\surnameend} \&
  \bibinfo{author}{Jorge~A \surnamestart P{\'e}rez\surnameend}
  (\bibinfo{year}{2017}): \emph{\bibinfo{title}{Linearity, control effects, and
  behavioral types}}.
\newblock In: {\sl \bibinfo{booktitle}{European Symposium on Programming}},
  \bibinfo{organization}{Springer}, pp. \bibinfo{pages}{229--259},
  \doi{10.1007/978-3-662-54434-1_9}.

\bibitemdeclare{inproceedings}{Caires10concur}
\bibitem{Caires10concur}
\bibinfo{author}{Lu{\'\i}s \surnamestart Caires\surnameend} \&
  \bibinfo{author}{Frank \surnamestart Pfenning\surnameend}
  (\bibinfo{year}{2010}): \emph{\bibinfo{title}{Session Types as Intuitionistic
  Linear Propositions}}.
\newblock In: {\sl \bibinfo{booktitle}{Proceedings of the 21st International
  Conference on Concurrency Theory (CONCUR 2010)}},
  \bibinfo{publisher}{Springer LNCS 6269}, \bibinfo{address}{Paris, France},
  pp. \bibinfo{pages}{222--236}, \doi{10.1007/978-3-642-15375-4_16}.

\bibitemdeclare{article}{Caires16mscs}
\bibitem{Caires16mscs}
\bibinfo{author}{Lu{\'\i}s \surnamestart Caires\surnameend},
  \bibinfo{author}{Frank \surnamestart Pfenning\surnameend} \&
  \bibinfo{author}{Bernardo \surnamestart Toninho\surnameend}
  (\bibinfo{year}{2016}): \emph{\bibinfo{title}{Linear Logic Propositions as
  Session Types}}.
\newblock {\sl \bibinfo{journal}{Mathematical Structures in Computer Science}}
  \bibinfo{volume}{26}(\bibinfo{number}{3}), pp. \bibinfo{pages}{367--423},
  \doi{10.1016/j.tcs.2010.01.028}.

\bibitemdeclare{article}{Cervesato09ic}
\bibitem{Cervesato09ic}
\bibinfo{author}{Iliano \surnamestart Cervesato\surnameend} \&
  \bibinfo{author}{Andre \surnamestart Scedrov\surnameend}
  (\bibinfo{year}{2009}): \emph{\bibinfo{title}{Relating State-Based and
  Process-Based Concurrency through Linear Logic}}.
\newblock {\sl \bibinfo{journal}{Information and Computation}}
  \bibinfo{volume}{207}(\bibinfo{number}{10}), pp. \bibinfo{pages}{1044--1077},
  \doi{10.1016/j.ic.2008.11.006}.

\bibitemdeclare{inproceedings}{DeYoung12csl}
\bibitem{DeYoung12csl}
\bibinfo{author}{Henry \surnamestart DeYoung\surnameend},
  \bibinfo{author}{Lu{\'\i}s \surnamestart Caires\surnameend},
  \bibinfo{author}{Frank \surnamestart Pfenning\surnameend} \&
  \bibinfo{author}{Bernardo \surnamestart Toninho\surnameend}
  (\bibinfo{year}{2012}): \emph{\bibinfo{title}{Cut Reduction in Linear Logic
  as Asynchronous Session-Typed Communication}}.
\newblock In \bibinfo{editor}{P.~\surnamestart C{\'e}gielski\surnameend} \&
  \bibinfo{editor}{A.~\surnamestart Durand\surnameend}, editors: {\sl
  \bibinfo{booktitle}{Proceedings of the 21st Conference on Computer Science
  Logic}}, \bibinfo{series}{CSL 2012}, pp. \bibinfo{pages}{228--242},
  \doi{10.4230/LIPIcs.CSL.2012.228}.

\bibitemdeclare{inproceedings}{Fowler19popl}
\bibitem{Fowler19popl}
\bibinfo{author}{Simon \surnamestart Fowler\surnameend}, \bibinfo{author}{Sam
  \surnamestart Lindley\surnameend}, \bibinfo{author}{J.~Garrett \surnamestart
  Morris\surnameend} \& \bibinfo{author}{S{\'a}ra \surnamestart
  Decova\surnameend} (\bibinfo{year}{2019}): \emph{\bibinfo{title}{Exceptional
  Asynchronous Session Types}}.
\newblock In: {\sl \bibinfo{booktitle}{Proceedings of the 46th Symposium on
  Programming Languages (POPL 2019)}}, \bibinfo{publisher}{ACM},
  \bibinfo{address}{Cascais, Portugal}, pp. \bibinfo{pages}{28:1--28:29}.

\bibitemdeclare{article}{Gay10jfp}
\bibitem{Gay10jfp}
\bibinfo{author}{Simon~J. \surnamestart Gay\surnameend} \&
  \bibinfo{author}{Vasco~T. \surnamestart Vasconcelos\surnameend}
  (\bibinfo{year}{2010}): \emph{\bibinfo{title}{Linear Type Theory for
  Asynchronous Session Types}}.
\newblock {\sl \bibinfo{journal}{Journal of Functional Programming}}
  \bibinfo{volume}{20}(\bibinfo{number}{1}), pp. \bibinfo{pages}{19--50},
  \doi{10.1006/inco.1994.1093}.

\bibitemdeclare{article}{Gentzen35}
\bibitem{Gentzen35}
\bibinfo{author}{Gerhard \surnamestart Gentzen\surnameend}
  (\bibinfo{year}{1935}): \emph{\bibinfo{title}{Untersuchungen {\"u}ber das
  Logische {S}chlie{\ss}en}}.
\newblock {\sl \bibinfo{journal}{Mathematische Zeitschrift}}
  \bibinfo{volume}{39}, pp. \bibinfo{pages}{176--210, 405--431},
  \doi{10.1007/BF01201353}.
\newblock \bibinfo{note}{English translation in M.~E. Szabo, editor, {\em The
  Collected Papers of Gerhard Gentzen}, pages 68--131, North-Holland, 1969}.

\bibitemdeclare{inproceedings}{Girard87tapsoft}
\bibitem{Girard87tapsoft}
\bibinfo{author}{J.-Y. \surnamestart Girard\surnameend} \&
  \bibinfo{author}{Y.~\surnamestart Lafont\surnameend} (\bibinfo{year}{1987}):
  \emph{\bibinfo{title}{Linear Logic and Lazy Computation}}.
\newblock In \bibinfo{editor}{H.~\surnamestart Ehrig\surnameend},
  \bibinfo{editor}{R.~\surnamestart Kowalski\surnameend},
  \bibinfo{editor}{G.~\surnamestart Levi\surnameend} \&
  \bibinfo{editor}{U.~\surnamestart Montanari\surnameend}, editors: {\sl
  \bibinfo{booktitle}{Proceedings of the International Joint Conference on
  Theory and Practice of Software Development}}, \bibinfo{volume}{2},
  \bibinfo{publisher}{Springer-Verlag LNCS 250}, \bibinfo{address}{Pisa,
  Italy}, pp. \bibinfo{pages}{52--66}, \doi{10.1007/BFb0014972}.

\bibitemdeclare{article}{Girard87tcs}
\bibitem{Girard87tcs}
\bibinfo{author}{Jean-Yves \surnamestart Girard\surnameend}
  (\bibinfo{year}{1987}): \emph{\bibinfo{title}{Linear Logic}}.
\newblock {\sl \bibinfo{journal}{Theoretical Computer Science}}
  \bibinfo{volume}{50}, pp. \bibinfo{pages}{1--102},
  \doi{10.1016/0304-3975(87)90045-4}.

\bibitemdeclare{phdthesis}{Griffith16phd}
\bibitem{Griffith16phd}
\bibinfo{author}{Dennis \surnamestart Griffith\surnameend}
  (\bibinfo{year}{2016}): \emph{\bibinfo{title}{Polarized Substructural Session
  Types}}.
\newblock Ph.D. thesis, \bibinfo{school}{University of Illinois at
  Urbana-Champaign}.

\bibitemdeclare{inproceedings}{Honda93concur}
\bibitem{Honda93concur}
\bibinfo{author}{Kohei \surnamestart Honda\surnameend} (\bibinfo{year}{1993}):
  \emph{\bibinfo{title}{Types for Dyadic Interaction}}.
\newblock In: {\sl \bibinfo{booktitle}{4th International Conference on
  Concurrency Theory}}, \bibinfo{series}{CONCUR'93},
  \bibinfo{publisher}{Springer LNCS 715}, pp. \bibinfo{pages}{509--523},
  \doi{10.1007/3-540-57208-2\_35}.

\bibitemdeclare{inproceedings}{Licata16lfcs}
\bibitem{Licata16lfcs}
\bibinfo{author}{Daniel~R. \surnamestart Licata\surnameend} \&
  \bibinfo{author}{Michael \surnamestart Shulman\surnameend}
  (\bibinfo{year}{2016}): \emph{\bibinfo{title}{Adjoint Logic with a 2-Category
  of Modes}}.
\newblock In: {\sl \bibinfo{booktitle}{International Symposium on Logical
  Foundations of Computer Science (LFCS)}}, \bibinfo{publisher}{Springer LNCS
  9537}, pp. \bibinfo{pages}{219--235}, \doi{10.1007/978-3-319-27683-0_16}.

\bibitemdeclare{inproceedings}{Licata17fscd}
\bibitem{Licata17fscd}
\bibinfo{author}{Daniel~R. \surnamestart Licata\surnameend},
  \bibinfo{author}{Michael \surnamestart Shulman\surnameend} \&
  \bibinfo{author}{Mitchell \surnamestart Riley\surnameend}
  (\bibinfo{year}{2017}): \emph{\bibinfo{title}{A Fibrational Framework for
  Substructural and Modal Logics}}.
\newblock In: {\sl \bibinfo{booktitle}{International Conference on Formal
  Structures for Computation and Deduction}}, \bibinfo{publisher}{LIPIcs},
  \bibinfo{address}{Oxford}, \doi{10.4230/LIPIcs.FSCD.2017.25}.

\bibitemdeclare{inproceedings}{Mostrous14coordination}
\bibitem{Mostrous14coordination}
\bibinfo{author}{Dimitris \surnamestart Mostrous\surnameend} \&
  \bibinfo{author}{Vasco \surnamestart Vasconcelos\surnameend}
  (\bibinfo{year}{2014}): \emph{\bibinfo{title}{Affine Sessions}}.
\newblock In \bibinfo{editor}{E.~\surnamestart K\"uhn\surnameend} \&
  \bibinfo{editor}{R.~\surnamestart Pugliese\surnameend}, editors: {\sl
  \bibinfo{booktitle}{16th International Conference on Coordination Models and
  Languages}}, \bibinfo{publisher}{Springer LNCS 8459},
  \bibinfo{address}{Berlin, Germany}, pp. \bibinfo{pages}{115--130},
  \doi{10.1007/978-3-662-43376-8_8}.

\bibitemdeclare{book}{Negri01book}
\bibitem{Negri01book}
\bibinfo{author}{Sara \surnamestart Negri\surnameend} \& \bibinfo{author}{Jan
  \surnamestart von Plato\surnameend} (\bibinfo{year}{2001}):
  \emph{\bibinfo{title}{Structural Proof Theory}}.
\newblock \bibinfo{publisher}{Cambridge University Press},
  \doi{10.1017/CBO9780511527340}.

\bibitemdeclare{article}{Padovani17icfp}
\bibitem{Padovani17icfp}
\bibinfo{author}{Luca \surnamestart Padovani\surnameend}
  (\bibinfo{year}{2017}): \emph{\bibinfo{title}{A Simple Library Implementation
  of Binary Sessions}}.
\newblock {\sl \bibinfo{journal}{Journal of Functional Programming}}
  \bibinfo{volume}{27}(\bibinfo{number}{e4}),
  \doi{10.1016/0304-3975(83)90059-2}.

\bibitemdeclare{unpublished}{Pfenning16course}
\bibitem{Pfenning16course}
\bibinfo{author}{Frank \surnamestart Pfenning\surnameend}
  (\bibinfo{year}{2016}): \emph{\bibinfo{title}{Law and Order}}.
\newblock
  \urlprefix\url{http://www.cs.cmu.edu/~fp/courses/15816-f16/lectures/08-lawandorder.pdf}.
\newblock \bibinfo{note}{Lecture notes on \textit{Substructural Logics}}.

\bibitemdeclare{inproceedings}{Pfenning15fossacs}
\bibitem{Pfenning15fossacs}
\bibinfo{author}{Frank \surnamestart Pfenning\surnameend} \&
  \bibinfo{author}{Dennis \surnamestart Griffith\surnameend}
  (\bibinfo{year}{2015}): \emph{\bibinfo{title}{Polarized Substructural Session
  Types}}.
\newblock In \bibinfo{editor}{A.~\surnamestart Pitts\surnameend}, editor: {\sl
  \bibinfo{booktitle}{Proceedings of the 18th International Conference on
  Foundations of Software Science and Computation Structures (FoSSaCS 2015)}},
  \bibinfo{publisher}{Springer LNCS 9034}, \bibinfo{address}{London, England},
  pp. \bibinfo{pages}{3--22}, \doi{10.1007/978-3-662-46678-0_1}.
\newblock \bibinfo{note}{Invited talk}.

\bibitemdeclare{misc}{Pfenning18tlla}
\bibitem{Pfenning18tlla}
\bibinfo{author}{Frank \surnamestart Pfenning\surnameend} \&
  \bibinfo{author}{Klaas \surnamestart Pruiksma\surnameend}
  (\bibinfo{year}{2018}): \emph{\bibinfo{title}{A Shared Memory Semantics for
  Session Types}}.
\newblock \bibinfo{howpublished}{Invited talk at the Workshop on
  Linearity/TLLA, Oxford, UK}.

\bibitemdeclare{unpublished}{Pruiksma18un}
\bibitem{Pruiksma18un}
\bibinfo{author}{Klaas \surnamestart Pruiksma\surnameend},
  \bibinfo{author}{William \surnamestart Chargin\surnameend},
  \bibinfo{author}{Frank \surnamestart Pfenning\surnameend} \&
  \bibinfo{author}{Jason \surnamestart Reed\surnameend} (\bibinfo{year}{2018}):
  \emph{\bibinfo{title}{Adjoint Logic}}.
\newblock \urlprefix\url{http://www.cs.cmu.edu/~fp/papers/adjoint18b.pdf}.
\newblock \bibinfo{note}{Unpublished manuscript}.

\bibitemdeclare{unpublished}{Reed09un}
\bibitem{Reed09un}
\bibinfo{author}{Jason \surnamestart Reed\surnameend} (\bibinfo{year}{2009}):
  \emph{\bibinfo{title}{A Judgmental Deconstruction of Modal Logic}}.
\newblock \urlprefix\url{http://www.cs.cmu.edu/~jcreed/papers/jdml2.pdf}.
\newblock \bibinfo{note}{Unpublished manuscript}.

\bibitemdeclare{inproceedings}{Scalas16ecoop}
\bibitem{Scalas16ecoop}
\bibinfo{author}{Alceste \surnamestart Scalas\surnameend} \&
  \bibinfo{author}{Nobuko \surnamestart Yoshida\surnameend}
  (\bibinfo{year}{2016}): \emph{\bibinfo{title}{Lightweight Session Programming
  in Scala}}.
\newblock In: {\sl \bibinfo{booktitle}{Proceedings of the 30th European
  Conference on Object-Oriented Programming (ECOOP 2016)}},
  \bibinfo{publisher}{LICIcs 56}, \bibinfo{address}{Rome, Italy}, pp.
  \bibinfo{pages}{21:1--21:28}, \doi{10.4230/LIPIcs.ECOOP.2016.21}.

\bibitemdeclare{inproceedings}{Toninho13esop}
\bibitem{Toninho13esop}
\bibinfo{author}{Bernardo \surnamestart Toninho\surnameend},
  \bibinfo{author}{Lu{\'\i}s \surnamestart Caires\surnameend} \&
  \bibinfo{author}{Frank \surnamestart Pfenning\surnameend}
  (\bibinfo{year}{2013}): \emph{\bibinfo{title}{Higher-Order Processes,
  Functions, and Sessions: A Monadic Integration}}.
\newblock In \bibinfo{editor}{\surnamestart M.Felleisen\surnameend} \&
  \bibinfo{editor}{\surnamestart P.Gardner\surnameend}, editors: {\sl
  \bibinfo{booktitle}{Proceedings of the European Symposium on Programming
  (ESOP'13)}}, \bibinfo{publisher}{Springer LNCS 7792}, \bibinfo{address}{Rome,
  Italy}, pp. \bibinfo{pages}{350--369}, \doi{10.1007/978-3-642-37036-6\_20}.

\bibitemdeclare{inproceedings}{Wadler12icfp}
\bibitem{Wadler12icfp}
\bibinfo{author}{Philip \surnamestart Wadler\surnameend}
  (\bibinfo{year}{2012}): \emph{\bibinfo{title}{Propositions as Sessions}}.
\newblock In: {\sl \bibinfo{booktitle}{Proceedings of the 17th International
  Conference on Functional Programming}}, \bibinfo{series}{ICFP 2012},
  \bibinfo{publisher}{ACM Press}, \bibinfo{address}{Copenhagen, Denmark}, pp.
  \bibinfo{pages}{273--286}, \doi{10.1145/2364527.2364568}.

\end{thebibliography}

\appendix

\section{Adjoint Logic}
\label{app:adjoint-logic}

We present here a brief overview of the formulation of adjoint logic that
we take as a basis for the semantics presented in the main body of the paper.
Adjoint logic can be thought of as a schema to define
particular logics.  The schema is parameterized by a set of modes of
truth $m$, where each proposition and logical connective is indexed by
its mode.  Furthermore, each mode intrinsically carries a set of
structural properties $\sigma(m) \subseteq \{\m{W},\m{C}\}$ where
$\m{W}$ stands for \emph{weakening} and $\m{C}$ stands for
\emph{contraction}.  As a concession to simplicity of the
presentation, in this paper we always allow exchange, although
nothing stands in the way of an even more general framework~\cite{Pfenning16course}.
In addition, an
instance requires a preorder between modes, where $m \geq k$ expresses
that the proof of a proposition of mode $k$ may depend on a hypotheses
of mode $m$.  This preorder embodies the \emph{declaration of
    independence}:
\begin{quote}\it
    A proof of $A_k$ may only depend on hypotheses $B_m$ for $m \geq k$.
\end{quote}
The form of a sequent is
\[
\Psi \vdash A_k \quad \mbox{where $\Psi \geq k$}
\]
where $\Psi$ is a collection of \emph{antecedents} of the form
$(x_i : B^i_{m_i})$ with each $m_i \geq k$, where all the variables
$x_i$ are distinct. This critical presupposition is abbreviated
as $\Psi \geq k$. Furthermore, the order of the antecedents does
not matter since we always allow exchange.

In addition, we require the preorder between modes to be compatible
with their structural properties: that is, $m \geq k$ implies
$\sigma(m) \supseteq \sigma(k)$. This is necessary to guarantee cut
elimination.

Finally, we may define fragments by restricting the set of
propositions we consider for a given mode.

The propositions at each mode are constructed uniformly, remaining
within the same mode, except for the \emph{shift operators} that move
between modes.  They are $\up_k^m A_k$ (pronounced \emph{up}), which is
a proposition at mode $m$ and requires $m \geq k$; and
$\down^\ell_m A_\ell$ (\emph{down}), which is also a proposition at mode
$m$, and which requires $\ell \geq m$.

At this point we can already write out the syntax of propositions.
\[
\begin{array}{lcl}
A_m, B_m & \Coloneqq & p_m \mid A_m \lolli_m B_m \mid A_m \tensor_m B_m \mid \one_m \mid \mathop{\oplus_m}\limits_{i \in I} A_m^i
\mid \mathop{\with_m}\limits_{i \in I} A_m^i \mid \up_k^m A_k \mid \down_m^\ell A_\ell
\end{array}
\]
Here $p_m$ stands for atomic propositions at mode $m$.  Due to the
needs of our operational interpretation, we generalize
internal and external choice to $n$-ary constructors parameterized by
an index set $I$.  So we write
$A_m^1 \oplus A_m^2 = \mathop{\oplus}\limits_{i \in \{1,2\}} A_m^i$.%
% \note{why \texttt{\textbackslash limits}? -wchargin}
% \note{not a big fan myself, but its a matter of taste -fp}

Remarkably, the right and left rules in the sequent calculus defining
the logical connectives are the same for each mode and are complemented
by the permissible structural rules.

\begin{figure*}
    \centering
    $\displaystyle
    \begin{array}{c}
    \infer[\id]
    {(x : A_m) \seq A_m}
    {\mathstrut}
    \qquad
    \infer[\cut]
    {\Psi \ \Psi' \seq C_k}
    {\Psi \geq m \geq k
        & \Psi \vdash A_m 
        & (x : A_m)\ \Psi' \vdash C_k
    }
    \\[1em]
    \infer[\weaken]
    {\Psi \ (x : A_m) \seq C_k}
    {\m{\m{W}} \in \sigma(m)
        &\Psi \seq C_k
    }
    \qquad
    \infer[\contract]
    {\Psi \ (x : A_m) \seq C_k}
    {\m{\m{C}} \in \sigma(m)
        &\Psi \ (y : A_m) \ (z : A_m) \seq C_k
    }
    \\[1em]
    \infer[{\oplus} R_\ell]
    {\Psi \seq \mathop{\oplus}\limits_{i \in I} A_m^i}
    {\ell \in I
        & \Psi \seq A_m^\ell}
    \qquad
    \infer[{\oplus} L]
    {\Psi \ (x : \mathop{\oplus}\limits_{i \in I} A_m^i) \seq C_k}
    {\Psi \ (y : A_m^i) \seq C_k \text{ for each $i \in I$}}
    \\[1em]
    \infer[{\with} R]
    {\Psi \seq \mathop{\with}\limits_{i \in I} A_m^i}
    {\Psi \seq A_m^i \text{ for each $i \in I$}}
    \qquad
    \infer[{\with} L_\ell]
    {\Psi\ (x :\mathop{\with}\limits_{i \in I} A_m^i) \seq C_k}
    {\ell \in I
        & \Psi\ (y : A_m^\ell) \seq C_k
    }
    \\[1em]
    \infer[{\otimes} R]
    {\Psi \ \Psi' \seq A_m \otimes B_m}
    {\Psi \seq A_m
        &\Psi' \seq B_m
    }
    \qquad
    \infer[{\otimes} L]
    {\Psi \ (x : A_m \otimes B_m) \seq C_k}
    {\Psi \ (y : A_m) \ (z : B_m) \seq C_k}
    \qquad\qquad
    \infer[{\one} R]
    {\cdot \seq \one_m}
    {\mathstrut}
    \qquad
    \infer[{\one} L]
    {\Psi \ (x : \one_m) \seq C_k}
    {\Psi \seq C_k}
    \\[1em]
    \infer[{\lolli} R]
    {\Psi \seq A_m \lolli B_m}
    {(x : A_m) \ \Psi \seq B_m}
    \qquad
    \infer[{\lolli} L]
    {\Psi \ \Psi' \ (x : A_m \lolli B_m) \seq C_k}
    {\Psi' \geq m
        &\Psi' \seq A_m
        &\Psi \ (y : B_m) \seq C_k
    }
    \\[1em]
    \infer[{\upshift} R]
    {\Psi \seq \upshift_k^m A_k}
    {\Psi \seq A_k}
    \qquad
    \infer[{\upshift} L]
    {\Psi \ (x : \upshift_k^m A_k) \seq C_\ell}
    {k \geq \ell
        &\Psi \ (y : A_k) \seq C_\ell
    }
    \qquad\qquad
    \infer[{\downshift} R]
    {\Psi \seq \downshift_k^m A_m}
    {\Psi \geq m
        &\Psi \seq A_m
    }
    \qquad
    \infer[{\downshift} L]
    {\Psi \ (x : \downshift_k^m A_m) \seq C_\ell}
    {\Psi \ (y : A_m) \seq C_\ell}
    \end{array}
    $
    \caption{Rules of Adjoint Logic}
    \label{fig:basic-logic-rules}
\end{figure*}

\subsection{Judgmental and structural rules}

The rules for adjoint logic can be found in \cref{fig:basic-logic-rules},
in which we give a more standard presentation of the logic than that used
by the operational semantics
(\cref{fig:message-passing-typing-rules}). We begin with the judgmental
rules of identity and cut, which express the connection between
antecedents and succedents.  Identity says that if we assume $A_m$ we
are allowed to conclude $A_m$.  Cut says the opposite: if we can
conclude $A_m$ we are allowed to assume $A_m$ \emph{as long as the
    declaration of independence is respected}.

As is common for the sequent calculus, we read the rules in the
direction of bottom-up proof construction.  For the cut rule, this
means we should assume that the conclusion $\Psi\ \Psi' \vdash C_k$
is well-formed and, in particular, that $\Psi \geq k$ and $\Psi' \geq k$.
Therefore, if we check that $m \geq k$, then we know that the second premise,
$(x : A_m)\ \Psi' \vdash C_k$, will also be well-formed.  For the first
premise to be well-formed, we need to check outright that $\Psi \geq m$.

\iffalse
\begin{example}[Cuts in LNL]
    In LNL~\cite{Benton94csl} we have two modes, $\mU > \mL$ with
    $\sigma(\mU) = \{\m{W},\m{C}\}$ and $\sigma(\mL) = \{\,\}$.  This one
    cut rule encompasses exactly the three cut rules in LNL and their
    restrictions via the declaration of independence:
    \begin{description}
        \item[$m = k = \mL$] corresponds to $\mathcal{L}\mathcal{L}\mbox{\it-cut}$
        with no further restriction.
        \item[$m = \mU$ and $k = \mL$] corresponds to $\mathcal{C}\mathcal{L}\mbox{\it-cut}$
        where $\Psi \geq \mU$ enforces that the first premise depends only
        on structural antecedents $B_\mU$.
        \item[$m = \mU$ and $k = \mU$] corresponds to $\mathcal{C}\mathcal{C}\mbox{\it-cut}$
        where $\Psi \geq \mU$ is already known by the well-formedness of the conclusion:
        both $\Psi$ and $\Psi'$ consist only of structural propositions $B_\mU$.
    \end{description}
\end{example}
\fi

% \subsection{Structural rules}

The structural rules of weakening and contraction just need to verify
that the mode of the principal formula permits the rule.

\subsection{Additive and multiplicative connectives}

The logical rules defining the additive and multiplicative
connectives are simply the linear rules for all modes, since
we have separated out the structural rules.  Except in
one case, ${\lolli}L$, the well-formedness of the conclusion
implies the well-formedness of all premises.

As for ${\lolli}L$, we know from the well-formedness of the conclusion
that $\Psi \geq k$, $\Psi' \geq k$, and $m \geq k$.  These facts by
themselves already imply the well-formedness of the second premise,
but we need to check that $\Psi' \geq m$ in order for the first premise
to be well-formed.

\subsection{Shifts}

The shifts represent the most interesting aspects of the rules.
Recall that in $\up^m_k A_k$ and $\down^m_k A_m$ we require that
$m \geq k$. We first consider the two rules for $\up$.
We know from the conclusion of the right rule that $\Psi \geq m$ and from
the requirement of the shift that $m \geq k$.  Therefore, as $\geq$ is 
transitive, $\Psi \geq k$ and the premise is always well-formed. This
also means (although we do not prove it here) that this rule is
\emph{invertible}.

From the conclusion of the left rule, we know $\Psi \geq \ell$,
$m \geq \ell$, and $m \geq k$.  This does not imply that
$k \geq \ell$, which we need for the premise to be well-formed and
thus needs to be checked.  Therefore, this rule is non-invertible.

The downshift rules are constructed analogously, taking only the
declaration of independence and properties of the preorder $\leq$ as
guidance. Note that in this case the left rule is always applicable
(that is, invertible), while the right rule is non-invertible.

\subsection{Multicut}
\label{ssec:multicut}

Because we have an explicit rule of contraction, cut elimination does
not follow by a simple structural induction.  However, we can follow
Gentzen~\cite{Gentzen35} and allow multiple copies of the same proposition to be
removed by the cut, which then allows a structural induction
argument. In anticipation of the operational interpretation, we have
labeled our antecedents with unique variables, so the generalized form
of cut called \emph{multicut} (see, for example, \cite{Negri01book})
can remove $n \geq 0$ copies.  Of course, such cuts are only legal if
the propositions that are removed satisfy the necessary structural
rules. For $n = 0$, we require that the mode $m$ support weakening.
\[
\infer[\cut(\emptyset)]
{\Psi \ \Psi' \seq C_k}
{\Psi \geq m \geq k
    & \m{W} \in \sigma(m)
    &\Psi \seq A_m
    &\Psi' \seq C_k
}
\]
For $n = 1$, we obtain the usual cut rule and no special requirements
are needed.
\[
\infer[\cut(\{x\})]
{\Psi \ \Psi' \seq C_k}
{\Psi \geq m \geq k
    & \Psi \seq A_m
    & (x : A_m) \ \Psi' \seq C_k
}
\]
For $n \geq 2$, the mode of the cut formula must admit contraction.
\[
\infer[\cut(S \cup \{x, y\})]
{\Psi \ \Psi' \seq C_k}
{\begin{array}[b]{l}
    \m{C} \in \sigma(m) \\
    \Psi \geq m \geq k
    \end{array}
    \; \Psi \seq A_m
    & (S \cup \{x,y\} : A_m) \ \Psi' \seq C_k
}
\]
Here, we have used the abbreviation
% NB(wchargin): In r249, Klaas changed the spacing on this line to add
% ~s everywhere. I am assuming that this was intended to improve the
% appearance given that the formula has to span two lines, given that
% (a) the extra-spaced version is used nowhere else and (b) the spacing
% is not exactly conventional. Instead, I have opted to reword the
% sentence slightly so that the standard spacing still looks good, improving the overall appearance of the line. If this assumption is in error, feel free to revert.
%$(\{x_1,~\ldots,~x_n\}~:~A_m) = (x_1~:~A_m) \ldots (x_n~:~A_m)$.
$(\{x_1, \dotsc, x_n\} : A_m)$ to stand for $(x_1 : A_m) \ldots (x_n : A_m)$.

Note that each of these rules has a side condition that can be
interpreted informally as stating that the number of antecedents cut
must be compatible with the mode $m$: if there are no
antecedents removed, $m$ must admit weakening, and if we remove two or
more, $m$ must admit contraction. This is exactly $|S| \sim m$ as defined
in \cref{sec:typing}.

This allows us to write down a single rule encompassing all three of the
above cases for multicut:
\[
\infer[\cut(S)]
{\Psi \ \Psi' \seq C_k}
{\Psi \geq m \geq k
    &|S| \sim m
    &\Psi \seq A_m
    & (S : A_m) \ \Psi' \seq C_k
}
\]
Note that the standard cut rule is the instance of the multicut rule where
$|S| = 1$, and so proving multicut elimination for adjoint logic also yields
cut elimination for the standard cut rule.

\subsection{Identity Expansion and Cut Elimination}

We present standard identity expansion and cut elimination results as
evidence for the correctness of the sequent calculus as capturing the
meaning of the logical connectives via their inference rules.
Cut-free proofs will always decompose propositions when read from
conclusion to premise and thus yield a conservative extension
result. Finally, the fine detail of the proof is significant because
the cut reductions, which constitute the essence of the proof, are
the basis for the operational semantics.

\begin{theorem}[Identity Expansion]
    If $\Psi \seq A_m$, then there exists a proof that $\Psi \seq A_m$
    using identity rules only at atomic propositions, which is cut-free
    if the original proof is.
\end{theorem}
\begin{proof}
    We begin by proving that for any formula $A_m$, there is a
    cut-free proof that $(x : A_m) \seq A_m$ using identity rules
    only at atomic propositions. This follows easily from an
    induction on $A_m$.
    
    Now, we arrive at the theorem by induction over the structure
    of the given proof that $\Psi \seq A_m$.
\end{proof}

\begin{theorem}[Cut Elimination]
    \label{thm:cut-elim}
    If $\Psi \seq A_m$, then there is a cut-free proof of $\Psi \seq A_m$.
\end{theorem}
\begin{proof}
This proof follows the structure of many cut-elimination results.
First we prove admissibility of multicut in the cut-free system.
This is established by a straightforward nested induction,
first on the proposition $A_m$ and then simultaneously on the structure
of the deductions $\DD$ and $\EE$.  This is followed by a simple structural induction to prove
cut elimination, using the admissibility of (multi)cut when it is
encountered.  If we ignore the modes, this proof is very similar to
the original proof of Gentzen~\cite{Gentzen35}.
\end{proof}

\begin{corollary}
    Adjoint logic is a conservative extension of each of the logics at a
    fixed mode. That is, if $\Psi \seq A_m$ is a sequent purely at mode
    $m$ (in that every type in $\Psi$ is at mode $m$ and neither $A_m$
    nor the types in $\Psi$ make use of shifts), then $\Psi \seq A_m$ is
    provable using the rules of adjoint logic iff it is provable using
    the rules which define the logic at mode $m$.
\end{corollary}

\subsection{Adjunction properties}
\label{ssec:adjunction}

As yet, we have not discussed the meaning of the name
``\kern-0.1em\textit{adjoint logic}''.  This can be justified by showing that for
fixed $k \le m$, $\downshift_k^m$ and $\upshift_k^m$ yield an
adjoint pair of functors $\downshift_k^m \dashv \upshift_k^m$.  Since
prior results (see~\cite{Benton94csl} and \cite{Licata17fscd})
already establish this property and we have little new to contribute
here, we omit the details here.

\section{Asynchronous Adjoint Logic}
\label{app:async-adjoint-logic}

As has been observed before, intuitionistic and classical linear
logics can be put into a Curry--Howard correspondence with
session-typed communicating
processes~\cite{Caires10concur,Wadler12icfp,Caires16mscs}.  A linear
logical proposition corresponds to a session type, and a sequent proof
to a process expression. The transition rules of the operational
semantics derive from the cut reductions.

Under the intuitionistic interpretation a sequent proof\footnote{for now on
    the linear fragment, and also labeling the succedent with a fresh
    variable} of
\[
(x_1 : A_\mL^1)\cdots(x_n : A_\mL^n) \vdash (x : A_\mL)
\]
corresponds to a process $P$ that \emph{provides} channel $x$ and uses
channels $x_i$.  The types of the channels prescribe the pattern of
communication: in the succedent, positive types
(${\oplus}, {\tensor}, {\one}$) will send and negative types
(${\with}, {\lolli}$) will receive.  In the antecedent, the roles are
reversed.  Cut corresponds to parallel composition of two processes,
with a private channel between them, while identity simply equates two
channels.

\subsection{Enforcing Asynchronous Communication}
\label{ssec:async}

Under this interpretation, a cut of a right rule against a matching
left rule allows computation to proceed by mimicking the cut reduction
from the proof of Theorem~\ref{thm:cut-elim}.  For example, a cut at
type $\mathop{\oplus}\limits_{i \in I} A_\mL^i$ is replaced by a cut at type
$A_\mL^\ell$ for some $\ell \in I$.  This corresponds to passing a message
(`$\ell$') from the process \emph{providing}
$x : \mathop{\oplus}\limits_{i \in I} A_\mL^i$ to the process \emph{using}
$x$.  By its very nature, this form of cut reduction is
\emph{synchronous}: both provider and client proceed simultaneously
because the channel $x : A_\ell$ connects the two process continuations.

For realistic languages, and also for the paradigm to smoothly extend
to the case of adjoint logic where some modes permit weakening and
contraction, we would like to prescribe \emph{asynchronous
    communication} instead.

We observe that the \emph{asynchronous $\pi$-calculus} replaces
the usual action prefix for output $x\langle y\rangle. P$ by a process
expression $x\langle y\rangle$ \emph{without a continuation}, thereby ensuring
that communication is asynchronous.  Such a process represents the
message $y$ sent along channel $x$.  Under our interpretation, the
continuation process corresponds to the proof of the premise of a
rule.  Therefore, if we can restructure the sequent calculus so that
the rules that send (${\oplus}R$, ${\one}R$, ${\tensor}R$, ${\down}R$,
${\with}L$, ${\lolli}L$, ${\up}L$) have zero premises, then we may
achieve a similar effect.

As an example, we consider the two right rules for ${\oplus}$.
Reformulated as axioms, they become
\[
\infer[{\oplus}R^0_1]
{A \vdash A \oplus B}
{} % {\mathstrut}
\hspace{3em}
\infer[{\oplus}R^0_2]
{B \vdash A \oplus B}
{} % {\mathstrut}
\]
In the presence of cut, these two rules together produce the same
theorems as the usual two right rules.  In one direction, we use
cut
\[
\infer[\m{cut}_A]
{\Delta \vdash A \oplus B}
{\Delta \vdash A
    & \infer[{\oplus}R^0_1]
    {A \vdash A \oplus B}
    {}}
\hspace{3em}
\infer[\m{cut}_B]
{\Delta \vdash A \oplus B}
{\Delta \vdash B
    & \infer[{\oplus}R^0_2]
    {B \vdash A \oplus B}
    {}}
\]
and in the other direction we use identity
\[
\infer[{\oplus}R_1]
{A \vdash A \oplus B}
{\infer[\m{id}_A]{A \vdash A}{}}
\hspace{3em}
\infer[{\oplus}R_2]
{B \vdash A \oplus B}
{\infer[\m{id}_B]{B \vdash B}{}}
\]
to derive the other rules.

Returning to the $\pi$-calculus, instead of explicitly \emph{sending}
a message $a\langle b\rangle.\, P$ we \emph{spawn} a new process in
parallel $a\langle b\rangle \mid P$.  This use of parallel composition
corresponds to a cut; receiving a message is achieved by cut reduction:
\[
\infer[\m{cut}_{A \oplus B}]
{\Delta', A \vdash C}
{\infer[{\oplus}R^0_1]{A \vdash A \oplus B}{\mathstrut}
    & \infer[{\oplus}L]
    {\Delta', A \oplus B \vdash C}
    {\deduce[Q_1]{\Delta', A \vdash C}{}
        & \deduce[Q_2]{\Delta', B \vdash C}{}}}
\quad \Longrightarrow \quad
\deduce[Q_1]{\Delta', A \vdash C}{}
\]
We see the cut reduction completely eliminates the cut in one step,
which corresponds precisely to receiving a message.  In this example
the message would be $\pi_1$ since the axiom ${\oplus}R^0_1$ was used;
for ${\oplus}R^0_2$ it would be $\pi_2$.

In summary, if we restructure the sequent calculus so that the
non-invertible rules (those that send) have zero premises, then (1)
messages are proofs of axioms, (2) message sends are modeled by cut,
and (3) message receives are a new form of cut reduction with a single
continuation.

In the process we give something up, namely the traditional cut
elimination theorem.  For example, the sequent
$\cdot \vdash \one \oplus \one$ has no cut-free proof since no rule
matches this conclusion.  The saving grace is that we can reach a
normal form where each cut just simulates the usual rules of the
sequent calculus.  This can be shown by translation to the ordinary
sequent calculus, applying cut elimination, and translating the result
back.  Proofs in this normal form have the subformula property.
Perhaps more importantly, we have session fidelity and deadlock
freedom (\cref{sec:metatheory}) for the corresponding process calculus
even in the presence of recursive types and processes, which is ultimately
what we care about for the resulting concurrent programming language.

\subsection{Eliminating Weakening and Contraction}
\label{ssec:structural-elimination}

We have introduced multicut entirely with the standard motivation of
providing a simple proof of the admissibility of cut using structural
induction. Surprisingly, we can streamline the system further by
using multicut to eliminate weakening and contraction from the logic
altogether, as in the system we use as the basis for our typing rules
(\cref{fig:message-passing-typing-rules}).

Consider a mode $m$ with $\m{C} \in \sigma(m)$.  Then
contraction is a simple instance of multicut with
an instance of the identity rule.
\[
\infer[\cut(\{y, z\})]
{\Psi \ (x : A_m) \seq C_k}
{\infer[\id]
    {(x : A_m) \seq A_m}
    {\mathstrut}
    & \Psi \ (y : A_m) \ (z : A_m) \seq C_k
}
\]
Similarly, for a mode $m$ with $\m{W} \in \sigma(m)$,
weakening is also an instance of multicut.
\[
\infer[\cut(\emptyset)]
{\Psi \ (x : A_m) \seq C_k}
{\infer[\id]
    {(x : A_m) \seq A_m}
    {\mathstrut}
    & \Psi \seq C_k}
\]
Cut reductions in the presence of contraction entail many residual
contractions, as is evident already from Gentzen's original proof.
Under our interpretation of contraction above, these residual
contractions simply become multicuts with the identity.  The
operational interpretation of identities then plays three related
roles: with one client, an identity achieves a renaming, redirecting
communication; with two or more clients, an identity implements
copying; with zero clients, its effect is cancellation or garbage collection.
The central role of identities can be seen in full detail in
Figure~\ref{fig:dynamics1}, once we have introduced our notation for
processes and process configurations.

\section{Program Examples}
\label{sec:example}

In the examples that follow, we will work with two modes, $\mL$ and $\mU$, with
$\mL < \mU$, $\sigma(\mL) = \emptyset$, and $\sigma(\mU)~=~\{W, C\}$.
In these examples we also use recursively defined types and processes
without formally defining these constructs, since they are well-known
from the literature and orthogonal to our concerns (see, for example,
\cite{Toninho13esop}).

\subsection{Example: Circuits}
\label{sec:circuits}

We call channels $c_\mU$ that are subject to weakening and contraction
\emph{shared channels}. As an example that requires shared channels we
use circuits.  We start by programming a nor gate that processes
infinite streams of zeros and ones.
\begin{tabbing}
    $\m{bits}^\infty_\mU = {\oplus}\{\m{b0} : \m{bits}^\infty_\mU, \m{b1} : \m{bits}^\infty_\mU\}$ \\[1ex]
    $x : \m{bits}^\infty_\mU, y : \m{bits}^\infty_\mU \vdash \mi{nor} :: (z : \m{bits}^\infty_\mU)$ \\[1ex]
    $z \leftarrow \mi{nor} \leftarrow x,y =$ \\
    \quad $\m{case}\, x\,$
    \= $(\,\m{b0}(x') \Rightarrow \m{case}\, y\,$ \= $(\, \m{b0}(y') \Rightarrow$
    \= $z' \leftarrow z.\m{b1}(z') \semi$ \\
    \>\>\> $z' \leftarrow \mi{nor} \leftarrow x',y'$ \\
    \>\> $\mid \m{b1}(y') \Rightarrow z' \leftarrow z.\m{b0}(z') \semi$ \\
    \>\>\> $z' \leftarrow \mi{nor} \leftarrow x',y'\,)$ \\
    \> $\mid \m{b1}(x') \Rightarrow \m{case}\, y\,$ \= $(\, \m{b0}(y') \Rightarrow$
    \= $z' \leftarrow z.\m{b0}(z') \semi$ \\
    \>\>\> $z' \leftarrow \mi{nor} \leftarrow x',y'$ \\
    \>\> $\mid \m{b1}(y') \Rightarrow z' \leftarrow z.\m{b0}(z') \semi$ \\
    \>\>\> $z' \leftarrow \mi{nor} \leftarrow x',y'\,)\,)$
\end{tabbing}
This is somewhat verbose, but note that all channels here are shared.
For this particular gate they could also be linear because they are
neither reused nor canceled.
This illustrates that programming can be uniform at different modes,
which is a significant advantage of our system over systems of session types
based on linear logic with an exponential ${!}A$.
Our implementation of $\mi{nor}$ has the property
that for bits $A$, $B$, and $C$ with $C = \lnot (A \lor B)$, the
following transitions are possible and characterize $\mi{nor}$:
\begin{tabbing}
    $\proc(\{a\}, \{a'\}, a, a.A(a')), \proc(\{b\}, \{b'\}, b, b.B(b')),
    \proc(S, \{a, b\}, c, c \leftarrow \mi{nor} \leftarrow a,b)$ \\
    $\null \longrightarrow^*
    \proc(c', \{a', b'\}, c', c' \leftarrow \mi{nor} \leftarrow a',b'),
    \proc(S, \{c'\}, c, c.C(c'))$ \quad ($c'$ fresh)
\end{tabbing}

This multi-step reduction is shown in full (one step at a time) below.
We only show the initial portion of each process term, which is enough to disambiguate
where in the program we are, as otherwise process terms become unwieldy
and reduce clarity. We also assume the existence of a rule $\m{call}$ that lets us invoke
a defined process, replacing the call with the process definition, after appropriate
substitution. At each step, we have highlighted in red the process(es) that are about to
transition.

\[
\begin{array}{rc}
\proc(\{a\}, \{a'\}, a, a.A(a')), \proc(\{b\}, \{b'\}, b, b.B(b')),
{\color{red} \proc(S, \{a, b\}, c, c \leftarrow \mi{nor} \leftarrow a,b)}
&\overset{\m{call}}{\Longrightarrow} \\[.5em]

{\color{red} \proc(\{a\}, \{a'\}, a, a.A(a'))}, \proc(\{b\}, \{b'\}, b, b.B(b')),
{\color{red} \proc(S, \{a, b\}, c, \m{case}\; a \ldots)}
&\overset{{\oplus}\; C}{\Longrightarrow} \\[.5em]

{\color{red} \proc(\{b\}, \{b'\}, b, b.B(b'))},
{\color{red} \proc(S, \{a', b\}, c, \m{case}\; b \ldots)}
&\overset{{\oplus}\; C}{\Longrightarrow} \\[.5em]

{\color{red} \proc(S, \{a', b'\}, c, z' \leftarrow \ldots)}
&\overset{\cut(\{z'\})}{\Longrightarrow} \\[.5em]

\proc(\{c'\}, \{a', b'\}, z', z' \leftarrow \mi{nor} \leftarrow a', b'), \proc(S, \{c'\}, c, c.C(c'))
\end{array}
\]

When we build an or-gate out of a nor-gate we need to exploit sharing
to implement simple negation. In the example below, $u$ and $u'$ are both
names for the same shared channel. The process invoked as
$\mi{nor} \leftarrow x, y$ will multicast a message to the clients of
$u$ and $u'$.
\begin{tabbing}
    $x : \m{bits}^\infty_\mU, y : \m{bits}^\infty_\mU \vdash \mi{or} :: (z : \m{bits}^\infty_\mU)$ \\[1ex]
    $z \leftarrow \mi{or} \leftarrow x,y =$ \\
    \quad $\{u, u'\} \leftarrow \mi{nor} \leftarrow x,y$ \\
    \quad $z \leftarrow \mi{nor} \leftarrow u,u'$
\end{tabbing}
An analogous computation to the above is possible, except that
at an intermediate stage of the computation, we will also have
a shared channel $d$ carrying the (multicast) message
$\proc(\{u, u'\}, \{d'\}, d, d.D(d'))$ with $D = \lnot(A \lor B)$.

\subsection{Example: Map}
\label{sec:map}

Mapping a process over a list allows us to demonstrate the use of replicable
services, as well as cancellation.
We define a whole family of types indexed by a type $A$, which is not
formally part of the language but is expressed at the metalevel.
\begin{tabbing}
    $\mi{list}_A = {\oplus}\{\m{cons} : A \tensor \mi{list}_A, \m{nil} : \one\}$
\end{tabbing}
Such a list should not be viewed as a data structure in memory.
Instead, it is a behavioral description of a stream of messages.
A process that maps a channel of type $A$ to one of type $B$ will
itself have type $A \lolli B$. However, this process must be
shared since it needs to be applied to every element.  We
therefore obtain the following type and definition, where all
channels not annotated with a mode subscript are at mode $\mL$.
\begin{tabbing}
    $\m{f}_\mU : \up_\mL^\mU(A_\mL \lolli B_\mL), l : \mi{list}_A \vdash \mi{map} :: (k : \mi{list}_B)$ \\
    $k \leftarrow \mi{map} \leftarrow f_\mU, l =$ \\
    \quad $\m{case}\, l\,$ \= $(\, \m{cons}(l') \Rightarrow$
    \= $\m{case}\; l' (\langle x, l''\rangle \Rightarrow$
    \hspace{6em}\= \% receive element $x:A$ with continuation $l''$ \\
    \>\> $\{f_\mU', f_\mU''\} \leftarrow (\nu a)a \leftarrow f_\mU$ \> \% duplicate the channel $f_\mU$ \\
    \>\> $f' \leftarrow f_\mU'.\m{shift}(f') \semi$ \> \% obtain a fresh linear instance $f'$ of $f_\mU'$  \\
    \>\> $y \leftarrow f'.\langle x,y\rangle \semi $ \> \% send $x$ to $f'$, response will be along fresh $y$ \\
    \>\> $k' \leftarrow k.\m{cons}(k') \semi$ \> \% select $\m{cons}$ \\
    \>\> $k'' \leftarrow k'.\langle y,k''\rangle \semi$ \> \% send $y$ with continuation $k''$ \\
    \>\> $k'' \leftarrow \mi{map} \leftarrow f_\mU'', l'')$ \> \% recurse with continuation channels \\
    \> $\mid \m{nil}(l') \Rightarrow$ \> $\emptyset \leftarrow (\nu a)a \leftarrow f_\mU$ \> \% Cancel the channel $f_\mU$ \\
    \>\> $k' \leftarrow k.\m{nil}(k') \semi$ \> \% select $\m{nil}$ \\
    \>\> $\m{case}\; l' (\langle \rangle \Rightarrow$ \> \% wait for $l'$ to close \\
    \>\> $k'.\langle\,\rangle \,))$ \> \% close $k'$ and terminate
\end{tabbing}
In this example, $f_\mU$ is a replicable and cancelable service. In the
case of a nonempty list, we create two names for the channel $f_\mU$ --- one
to use immediately and one to pass to the recursive call. Note that the service
itself remains a single service with two clients until the message $\m{shift}(f')$
is sent to it, at which point it replicates itself, creating one copy to handle
this request and leaving another to deal with future requests. In the case of an
empty list, we have no elements to map over, and so we do not need to use $f_\mU$.
As such, we cancel it before continuing.

\end{document}